\documentclass[a4paper,11pt,onecolumn,notitlepage]{article}

\usepackage[pagebackref=true,pdfpagemode=UseNone,bookmarksopen=false,colorlinks=true,linkcolor=blue,citecolor=blue,pdftitle={Lueders' and quantum Jeffrey's rules as entropic projections},
pdfauthor={Ryszard Pawel Kostecki}]{hyperref}
\usepackage[T1]{fontenc}
\usepackage[all]{xy}
\usepackage{amsmath}
\usepackage{amssymb}
\usepackage{amsthm}
\usepackage{bibspacing}
\usepackage{rpk}
\renewcommand*{\backref}[1]{}
\renewcommand*{\backrefalt}[4]{%
  \ifcase #1 %
    \relax
  \or
    $\uparrow$~#2.
  \else
    $\uparrow$~#2.
  \fi%
}

\begin{document}
\thispagestyle{empty}
\numberwithin{equation}{section}
\renewcommand*{\thefootnote}{\fnsymbol{footnote}}

\begin{center}
\noindent
{\Large \textbf{L\"{u}ders' and quantum Jeffrey's rules as entropic projections}}\\
\ \\
{
Ryszard Pawe{\l} Kostecki}\\
{\scriptsize \ \\Perimeter Institute for Theoretical Physics, 31 Caroline Street North, N2L 2Y5 Waterloo, Ontario, Canada\footnote{Current affiliation.}\\Institute for Theoretical Physics, Department of Physics, University of Warsaw, Ho\.za 69, 00-681 Warszawa, Poland\\\vspace{1mm}\texttt{ryszard.kostecki@fuw.edu.pl}}\\\ 
{\scriptsize\ \\}{\small August 15, 2014}
\end{center}
\begin{abstract}
\noindent We prove that the standard quantum mechanical description of a quantum state change due to measurement, given by L\"{u}ders' rules, is a special case of the constrained maximisation of a~quantum relative entropy functional. This result is a~quantum analogue of the derivation of the Bayes--Laplace rule as a~special case of the constrained maximisation of relative entropy. The proof is provided for the Umegaki relative entropy of density operators over a Hilbert space as well as for the Araki relative entropy of normal states over a $W^*$-algebra. We also introduce a quantum analogue of Jeffrey's rule, derive it in the same way as above, and discuss the meaning of these results for quantum bayesianism.
\end{abstract}

\setcounter{footnote}{0}
\renewcommand*{\thefootnote}{\arabic{footnote}}
\section{Introduction\label{introduction.section}}
An important part of a mathematical setting of probability theory can be considered as a special case of quantum theoretic kinematics. This can be seen most clearly when quantum theoretic kinematics is reformulated in algebraic terms, with quantum states defined as normal positive (or normalised) functionals over $W^*$-algebras (see e.g. \cite{Redei:Summers:2007,Kostecki:2013} for an overview). The Borel--Steinhaus--Kolmogorov measure theoretic approach to foundations of probability theory is then recovered precisely from commutative $W^*$-algebras and quantum states over them.\footnote{The same holds for the Kappos' approach to foundations based on Caratheodory's measure theory on abstract boolean algebras, as well for Le~Cam's approach based on abstract Banach lattices, if the latter is restricted to the `coherent' models. In general, the semi-finite measures are recovered not from quantum states, but from semi-finite normal weights (see Section \ref{math.section}).} However, both quantum theory and probability theory are equipped with additional structures, which describe possible mappings of probabilities or quantum states. Can these prescriptions of information dynamics also be directly related to each other? In particular, there are various considerations \cite{Bub:1977,Bub:1979,Bub:1979:measurement,Redei:1992:Bayes,Caves:Fuchs:Schack:2001,Schack:Brun:Caves:2001,Fuchs:2002,Fuchs:2003,Jacobs:2002,Valente:2007,Streater:2007,Bub:2007,Palge:Konrad:2008,Henderson:2010} of von Neumann's and L\"{u}ders' rules \cite{vonNeumann:1932:grundlagen,Lueders:1951} as noncommutative \textit{analogues} of the Bayes--Laplace rule \cite{Bayes:1763,Laplace:1774,Laplace:1814}. It is tempting to ask whether this analogy could be turned to something more definite, both conceptually and mathematically. 

The main motivation for this paper is a series of results \cite{May:Harper:1976,May:1976,May:1979,Williams:1980,Diaconis:Zabell:1982,Domotor:1985,Zellner:1988,Warmuth:2005,Caticha:Giffin:2006,Giffin:2008,Douven:Romeijn:2012} showing that both Bayes--Laplace and Jeffrey's rules (as well as some other rules \cite{Domotor:Zanotti:Graves:1980,Hughes:vanFraassen:1984,Douven:Romeijn:2012}) can be derived as special cases of the constrained minimisation of various information distances on probabilistic models (in Section \ref{Bayes.section} we review briefly some of these results). This has led us to conjecture \cite{Kostecki:2011:Waterloo:talk} that L\"{u}ders' rules may be special cases of constrained maximisation of a quantum relative entropy. In Section \ref{vN.Lu.section} we prove that the weak L\"{u}ders rule of quantum state change due to ``nonselective quantum measurement'' is a special case of quantum entropic projection, provided by the minimisation of the Araki distance subject to a specific set of constraints. We also introduce a quantum analogue of Jeffrey's rule and derive it as another special case of constrained Araki distance minimisation. These are the two main results of this paper. In addition, we show that the strong L\"{u}ders rule of quantum state change due to ``selective quantum measurement'' can be obtained from these results as the limiting case of quantum Jeffrey's rule or by regularised Araki distance minimisation. In Section \ref{transition.correlation.section} we show that both weak L\"{u}ders' rule and strong von Neumann's rule (which is the same as strong L\"{u}ders' rule for pure states) are also quantum entropic projections for a different quantum distance functional. With an exception of a derivation of a quantum Jeffrey's rule, our results hold for arbitrary $W^*$-algebras, so they are applicable in quantum field theoretic and relativistic quantum information problems. 

These results extend earlier considerations of L\"{u}ders' rules as \textit{analogues} of the Bayes--Laplace rule with a novel mathematical and conceptual content: all these rules are special cases of the constrained relative entropic inference. In this sense, `quantum bayesianism' can be considered a branch of `quantum relative entropism'. A discussion of this issue is provided in Section \ref{discussion.section}.

\ \\\noindent\textbf{History of the problem.} An inference based on minimisation of quantum distance on quantum models was first proposed by Herbut \cite{Herbut:1969}. He derived the weak L\"{u}ders rule from a constrained minimisation of norm distance in the Hilbert--Schmidt operator space. Unfortunately, his work has been left unnoticed by all works cited below. Several years later Marchand and collaborators \cite{Marchand:Wyss:1977,Benoist:Marchand:Yourgrau:1977,Benoist:Marchand:1979,Benoist:Marchand:Wyss:1979,Gudder:Marchand:Wyss:1979,Gudder:1980,Marchand:1981,Marchand:1983,Marchand:1983:Milano} used Bures' distance \cite{Bures:1969} (which is metrical), and argued that its constrained minimisation should be considered as a description of the change of state of information \textit{due to} ``quantum measurement'' described by a specific form of a coarse graining map \cite{Marchand:1977} (conditioned upon a subset of an operator algebra, and predual to a specific form of noncommutative conditional expectation \cite{Gudder:Marchand:1977}). Independently of this body of work, in \cite{Carazza:Casartelli:DElia:1977} it was proposed to use the constrained minimisation of the WGKL distance \eqref{WGKL.distance} of probabilities arising from traces of density operators to derive the post-measurement quantum state. The rules of inferential change of quantum states based on a constrained minimisation of other metrical distances on quantum models were later reconsidered by other authors \cite{Hadjisavvas:1978,Hadjisavvas:1981,Dieks:Veltkamp:1983,Raggio:1984} and some derivations of the strong von Neumann rule were obtained (see Section \ref{transition.correlation.section}). The reinterpretation of von Neumann's and L\"{u}ders' rules for ``quantum measurement'' as principles of inductive inference conditioned on specific information, and analogous to the Bayes--Laplace rule, was proposed at about the same time by Bub \cite{Bub:1977,Bub:1979,Bub:1979:measurement} (however, it can be claimed \cite{Streater:2007}, that already von Neumann was aware of the possibility of such interpretation). These two lines of thought were (implicitly) joined in Hadjisavvas' \cite{Hadjisavvas:1978,Hadjisavvas:1981} postulate that a quantum state change due to acquisition of data (e.g. in a measurement) should be provided by means of constrained minimisation of the JMGK distance \eqref{D.JMGK}, as well as in Donald's \cite{Donald:1986,Donald:1987:further:results} postulate that a description of ``quantum measurement'', understood as an inductive inference, should be provided by means of constrained minimisation of the Araki distance.\footnote{\cytat{Leaving aside any possible applications in the area of quantum communication theory, it seems to me that, almost regardless of the interpretation
one places on quantum mechanics, this is an appropriate way of modelling the quantum measurement process} \cite{Donald:1987:further:results}. He states this refering to a constrained minimisation of a different distance functional, which however coincides with Araki's distance $D_1|_{\N^+_{\star1}}$ at least for injective $W^*$-algebras $\N$.} Donald stressed also that this procedure \cytat{clearly allows for approximate measurement, and indeed they are required (all real measurements are approximate), if $K$ is taken to have a non-empty interior} \cite{Donald:1987:further:results}, where $K\subseteq\Scal(\N)$ is the constraint set. However, he provided no derivation of any of L\"{u}ders' rules (nor any other ``quantum measurement'' rules) from this procedure. Independently of the above works, Warmuth \cite{Warmuth:2005} used constrained minimisation of Umegaki's distance to derive a generalisation of the Bayes--Laplace rule to the case of density operators, with conditional probabilities replaced by covariance matrices. This generalisation has not reproduced L\"{u}ders' rules.

Our derivation of L\"{u}ders' rule is not only the first such result obtained for Umegaki's and Araki's distances, but also first result of this type obtained for any nonsymmetric quantum information distance. All results for quantum Jeffrey's rule are new. See also a closely related paper \cite{HKK:2014}, where the analogous results for weak and strong L\"{u}ders' rules are derived using another technique (based on differentiation, as opposed to generalised pythagorean theorem).\footnote{After finishing this paper, we were informed about reference \cite{MPSVW:2010}, where it is shown that $\sigma=\sum_iP_i\rho_iP_i$, where $P_i$ are rank $1$ projectors, minimises the Umegaki distance $D_1|_{\BH^+_{\star1}}(\rho,\sigma)$. This is a special case of our result for the weak L\"{u}ders rule. The generalisation to our result is stated without proof in \cite{Coles:2012}.}
\section{Bayes--Laplace and Jeffrey's rules as entropic projections\label{Bayes.section}}
The modern mathematical formulation of the Bayes--Laplace foundations for probability theory \cite{Bayes:1763,Laplace:1812} is based on finitely additive boolean algebras $\boole$ and conditional probabilities, defined as maps $p(\cdot|\cdot):\boole\times\boole\ni(x,y)\mapsto p(x|y)\in[0,1]$. Its kinematics is given by the rules
\begin{align}
p(x|y)+p(\lnot x|y)&=1,\\
p(x|y\land z)p(y|z)&=p(x\land y|z),
\end{align}
from which the \df{Bayes--Laplace theorem} \cite{Bayes:1763,Laplace:1774,Laplace:1814} follows,
\begin{equation}
	p(x|b\land\eta)=p(x|\eta)\frac{p(b|x\land\eta)}{p(b|\eta)}.
\end{equation}
The `marginal probability' $p(b|\eta)$, called also an `evidence', is a normalising constant calculated from 
\begin{equation}
	p(b|\eta)=\sum_{i\in I} p(b|x_i\land\eta)p(x_i|\eta),
\end{equation}
where $I$ is a countable set, while the set $\{x_i\in\boole\mid i\in I\}$ is exhaustive ($\bigvee_{i\in I}x_i=1$) and its elements are mutually exclusive ($x_i\land x_j=0$ for $i\neq j$). In the simple cases this set may consist of two elements: $\{x,\lnot x\}$. The dynamics of this approach is given by the \df{Bayes--Laplace rule} 
\begin{equation}
	p(x|\eta)\mapsto p_{\mathrm{new}}(x|\eta):=p(x|\eta)\frac{p(b|x\land \eta)}{p(b|\eta)}.
\label{Bayes.rule.eq}
\end{equation}
The map \eqref{Bayes.rule.eq} determines a rule of construction of the new conditional probability assignment associated to $x$ under the constraint that certain additional statements (`facts', `data', `events') $b$ are considered as (`appear as', `are known as', `are regarded as') true or false. So, $p_{\mathrm{new}}(x|\eta)$ in \eqref{Bayes.rule.eq} is interpreted as a (`posterior') conditional probability assigned to $x$ whenever the truth value of $b$ is given (`known'). If attribution of a definite truth value to $b\in\boole$ is interpreted as an `acquisition of data/facts', then the Bayes--Laplace rule can be understood as a procedure of statistical inference that transforms `prior' information states $p(x|\eta)$ about \textit{all} `hypotheses' $x\in\boole$ into `posterior' information states $p_{\mathrm{new}}(x|\eta)$, under constraints provided by the acquired data $b$ \textit{and} the assumed relationship between $b$ and $x$ which is encoded in the information states $p(b|x\land\eta)$. The probability $p(b|x\land\eta)$ is sometimes called the `sampling probability' (because before the constraint $b$ is applied, $p(b|x\land\eta)$ can represent a probability associated with possible values of constraints for a fixed hypothesis $x$) or the `likelihood' (because after the constraint $b$ is applied, $p(b|x\land\eta)$ is considered as a probability of $b$ as a function over possible hypotheses $x$). 

Jeffrey \cite{Jeffrey:1957,Jeffrey:1965,Jeffrey:1968} proposed an alternative dynamical rule, now called \df{Jeffrey's rule},
\begin{equation}
	p(x|\eta)\mapsto p_{\mathrm{new}}(x|\eta):=\sum_{i=1}^np(x|b_i)\lambda_i=\sum_{i=1}^n\frac{p(x\land b_i|\eta)}{p(b_i|\eta)}\lambda_i,
\label{Jeffrey.rule}
\end{equation}
where $n\in\NN$, $\{b_1,\ldots,b_n\}$ is a set of exhaustive and mutually exclusive elements of $\boole$, and the constraints $\lambda_i=p_{\mathrm{new}}(b_i|\eta)$ $\forall i\in\{1,\ldots,n\}$ hold. The defining equality in \eqref{Jeffrey.rule} is equivalent to the condition
\begin{equation}
	p_{\mathrm{new}}(x|b_i)=p(x|b_i)\;\;\;\forall x\in\boole\;\;\forall i\in\{1,\ldots,n\}.
\label{Jeffrey.sufficiency.eq}
\end{equation}
If $n=2$ with $b_2=\lnot b_1=:\lnot b$ and if $p_{\mathrm{new}}(\lnot b|\eta)=0$, then Jeffrey's rule \eqref{Jeffrey.rule} reduces to the Bayes--Laplace rule \eqref{Bayes.rule.eq}.

The first derivation of the Bayes--Laplace and Jeffrey's rules from constrained minimisation of the information distance functional (more precisely, the WGKL distance) was provided by Williams \cite{Williams:1980}. However, as he admitted, his derivation has a restricted validity, because it does not deal with conditional probabilities. This can be solved following van Fraassen's suggestion: \cytat{When two spaces are used (parameter or hypothesis space and observation or sample space) I shall think of these as subspaces of a larger one (possibly produced by a product construction), so that in a single context all propositions are represented by measurable sets in a single space} \cite{vanFraassen:1981}. Following the results of Caticha and Giffin \cite{Caticha:Giffin:2006,Giffin:2008}, consider a constrained minimisation of the WGKL distance \eqref{WGKL.distance},
\begin{equation}
	p(\xx,\theta)\mapsto p_{\mathrm{new}}(\xx,\theta):=\arginf_{q(\xx,\theta)\in\M}\left\{\int_\X\tmu(\xx)q(\xx,\theta)\log\left(\frac{q(\xx,\theta)}{p(\xx,\theta)}\right)+F(q(\xx,\theta))\right\},
\label{caticha.giffin.rule}
\end{equation}
for $p,q\in\M:=\M(\X,\mho(\X),\tmu)\subseteq L_1(\X,\mho(\X),\tmu)_1^+$, $\dim\M=:n<\infty$, with parametrisation $\theta:\M(\X,\mho(\X),\tmu)\ra\Theta\subseteq\RR^n$ allowing to consider a measure space $(\Theta,\mho_{\mathrm{Borel}}(\Theta),\dd\theta)$ as well as a product measure space $(\X\times\Theta,\mho(\X\times\Theta),\tmu\times\dd\theta)$, and with constraints given by 
\begin{equation}
	F(q(\xx,\theta))=\lambda_1\left(\int_\X\tmu(\xx)\int_\Theta\dd\theta q(\xx,\theta)-1\right)+\lambda_2(\xx)\left(\int_\Theta\dd\theta q(\xx,\theta)-\dirac(\xx-b)\right),
\label{caticha.giffin.bayes.constraints}
\end{equation}
where $\lambda_1$ and $\lambda_2(\xx)$ are Lagrange multipliers, and $\dirac(\xx-b)$ is Dirac's delta at $b\in\X$. The posterior probability selected as a unique solution of this variational problem is given by
\begin{equation}      
        p_{\mathrm{new}}(\xx,\theta)=\frac{p(\xx,\theta)\ee^{\lambda_2(\xx)}}{\int_\X\tmu(\xx)\int_\Theta\dd\theta p(\xx,\theta)\ee^{\lambda_2(\xx)}},
\end{equation}
where $\lambda_2(\xx)$ is determined via
\begin{equation}
        \frac{\int_\Theta\dd\theta p(\xx,\theta)\ee^{\lambda_2(\xx)}}{\int_\X\tmu(\xx)\int_\Theta\dd\theta p(\xx,\theta)\ee^{\lambda_2(\xx)}}=\dirac(\xx-b).
\end{equation}
Hence,
\[
        p_{\mathrm{new}}(\xx,\theta)=\displaystyle\frac{p(\xx,\theta)\dirac(\xx-b)}{\int_\Theta\dd\theta p(\xx,\theta)}=\displaystyle\frac{p(\xx,\theta)\dirac(\xx-b)}{p(\xx)}=:\dirac(\xx-b)p(\theta|\xx),
\]
which leads to the Bayes--Laplace rule \eqref{Bayes.rule.eq} on $\Theta$,\footnote{More precisely, it leads to a generalisation of this rule to a domain of countably additive boolean algebras (representable by the Loomis--Sikorski theorem), which contains contains the finitary rule \eqref{Bayes.rule.eq} as a special case. While there is no universal (generic) extension of the notion of conditional probability to a measure theoretic framework, $p(\theta|\xx)$ can be interpreted as a density of a conditional measure under some conditions. 
This is not problematic as long as one interprets the above derivation as a proof of backwards compatibility of entropic projections with the Bayes--Laplace framework (which is finitary anyway), and not as a method of extending this framework from finite to countably additive boolean algebras.}
\begin{equation}
        p(\theta)\mapsto p_{\mathrm{new}}(\theta)=\int_\X\tmu(\xx)\dirac(\xx-b)p(\theta|\xx)=p(\theta|b),
\label{Bayes.from.MRE} 
\end{equation}
whenever $\mu$ is such that $\int_\X\tmu(\xx)\dirac(\xx-b)h(\xx)=h(b)$ (for example, if $\tmu(\xx)=\dd\xx$). If the second constraint in \eqref{caticha.giffin.bayes.constraints} is replaced by a more general form,
\begin{equation}
	F(q(\xx,\theta))=\lambda_1\left(\int_\X\tmu(\xx)\int_\Theta\dd\theta q(\xx,\theta)-1\right)+\lambda_2(\xx)\left(\int_\Theta\dd\theta q(\xx,\theta)-f(\xx)\right),
\end{equation}
corresponding to a condition $q(\xx)=\int_\Theta\dd\theta q(\xx,\theta)=f(\xx)$ with a given probability density $f\in\M(\X,\mho(\X),\tmu)$, then the entropic projection \eqref{caticha.giffin.rule} reproduces Jeffrey's rule \eqref{Jeffrey.rule} on $\Theta$,
\begin{align}
p_{\mathrm{new}}(x,\theta)&=\frac{p(\xx,\theta)}{p(\xx)}f(\xx)=:p(\xx|\theta)f(\xx)=p(\xx|\theta)p_{\mathrm{new}}(\xx),\\
p(\theta)\mapsto p_{\mathrm{new}}(\theta)&=\int_\X\tmu(\xx)f(\xx)\frac{p(\xx,\theta)}{p(\xx)}=\int_\X\tmu(\xx)p(\theta|\xx)f(\xx)=\int_\X\tmu(\xx)p(\theta|\xx)p_{\mathrm{new}}(\xx).
\end{align}

The Bayes--Laplace rule changes information states by means of constraints imposed on the level of propositions, while the entropic projections utilise constraints imposed on probabilities. Hence, in order to recover the former rule from the maximum entropy rule one needs a very strong constraint, which forces a unique reference of probability distribution to an underlying space $\X$ of propositions: $\int_\Theta\dd\theta p_{\mathrm{new}}(\xx,\theta)=\dirac(\xx-b)$. So, while (for example) the mean value constraints can be partially dismissed by the new knowledge that is incorporated by the sequential maximum relative entropy updating (see e.g. \cite{Tribus:Rossi:1973,Giffin:Caticha:2006}), Dirac's delta constraints always remain preserved by subsequent updatings.\footnote{In the above discussion Dirac's delta constraint is applied by integration over $\Theta$ space, while usually the mean value constraints are applied by integration over $\X$ space, but nevertheless this remark holds in general.} From this perspective, Jeffrey's rule can be understood as arising due to weakening of constraints, which are allowed to carry some additional uncertainty.

More generally, Diaconis and Zabell \cite{Diaconis:Zabell:1982} have shown that, for a suitable choice of constraints, Jeffrey's rule can be derived from a constrained minimisation of any Csisz\'{a}r--Morimoto $\fff$-distance \cite{Csiszar:1963,Morimoto:1963,Ali:Silvey:1966} for a strictly convex $\fff:\,]0,\infty[\,\ra\RR$. In particular, \cite{Douven:Romeijn:2012} derived Jeffrey's rule from constrained minimisation of $D_0|_{L_1(\X,\mho(\X),\tmu)^+_1}$. In Theorem \ref{weighted.vNLu.theorem} of Section \ref{vN.Lu.section} we will derive a quantum analogue of the latter result.
\section{Quantum distances and relative modular theory\label{math.section}}
In this section we present the mathematical terminology that is used throughout the paper. We also briefly introduce some more advanced notions and results from the theory of operator algebras that are required to set up the mathematical background and notation for the $W^*$-algebraic part of the results presented in Sections \ref{vN.Lu.section} and \ref{transition.correlation.section}. See \cite{Kostecki:2013} for a detailed discussion of these structures and their properties.
\subsection{States and weights over $W^*$-algebras}
A \df{$C^*$-algebra} is a Banach space $\C$ over $\CC$ with unit $\II$ that is also an algebra over $\CC$ and is equipped with an operation $^*:\C\ra\C$ satisfying $(xy)^*=y^*x^*$, $(x+y)^*=x^*+y^*$, $x^*{}^*=x$, $(\lambda x)^*=\lambda^*x^*$, and $\n{x^*x}=\n{x}^2$, where $\lambda^*$ is a complex conjugation of $\lambda\in\CC$. A \df{$W^*$-algebra} is defined as such $C^*$-algebra that has a Banach predual. If a predual of $C^*$-algebra exists then it is unique. Given a $W^*$-algebra $\N$, we will denote its predual by $\N_\star$. Moreover, $\N_\star^+:=\{\phi\in\N_\star\mid\phi(x^*x)\geq0\;\forall x\in\N\}$, $\N^+_{\star0}:=\{\phi\in\N^+_\star\mid\omega(x^*x)=0\limp x=0\;\;\forall x\in\N\}$, $\N^+_{\star1}:=\{\phi\in\N^+_\star\mid\n{\phi}=1\}$, $\N^\sa:=\{x\in\N\mid x^*=x\}$, $\N^+:=\{x\in\N\mid\exists y\in\N\;\;x=y^*y\}$, $\Proj(\N):=\{x\in\N^\sa\mid xx=x\}$. For $\N=\BH$, $\N_\star=\schatten_1(\H):=\{x\in\BH\mid\n{x}_{\schatten_1(\H)}:=\tr(\sqrt{x^*x})<\infty\}$. If $(\X,\mho(\X),\tmu)$ is a localisable measure space, then $L_\infty(\X,\mho(\X),\tmu)$ is a commutative $W^*$-algebra, and $L_1(\X,\mho(\X),\tmu)$ is its predual. Every commutative $W^*$-algebra can be represented in this form. This way the theory of $W^*$-algebras generalises both the localisable measure theory and the theory of bounded operators over Hilbert spaces. We define a \df{statistical model} as a set $\M(\X,\mho(\X),\tmu)\subseteq L_1(\X,\mho(\X),\tmu)^+$, where $(\X,\mho(\X),\tmu)$ is a localisable measure space. The elements of $L_1(\X,\mho(\X),\tmu)^+_1$ are Radon--Nikod\'{y}m quotients of probability measures (dominated by $\tmu$) with respect to $\tmu$, and are called \df{probability densities}. If $\N$ is a $W^*$-algebra, then we define a \df{quantum model} as a set $\M(\N)\subseteq\N^+_{\star}$. The elements of $\N^+_{\star}$ will be called \df{quantum states} or (just) \df{states}.

A \df{weight} on a $W^*$-algebra $\N$ is defined as a function $\omega:\N^+\ra[0,+\infty]$ such that $\omega(0)=0$, $\omega(x+y)=\omega(x)+\omega(y)$, and $\lambda\geq0\limp\omega(\lambda x)=\lambda\omega(x)$, with the convention $0\cdot(+\infty)=0$. A weight is called: \df{faithful} if{}f $\omega(x)=0\limp x=0$; \df{finite} if{}f $\omega(\II)<\infty$; \df{semi-finite} if{}f a left ideal in $\N$ given by
\begin{equation}\rpktarget{nnn}
        \nnn_\phi:=\{x\in\N\mid\phi(x^*x)<\infty\}
\end{equation}
is weakly-$\star$ dense in $\N$; \df{trace} if{}f $\omega(xx^*)=\omega(x^*x)\;\forall x\in\N$; \df{normal} if{}f $\omega(\sup\{x_\iota\})=\sup\{\omega(x_\iota)\}$ for any uniformly bounded increasing net $\{x_\iota\}\subseteq\N^+$. A space of all normal semi-finite weights on a $W^*$-algebra $\N$ is denoted $\W(\N)$, while the subset of all faithful elements of $\W(\N)$ is denoted $\W_0(\N)$. Every state is a finite normal weight, and every faithful state is a finite faithful normal state, hence the diagram
\begin{equation}
\xymatrix{
        \N^+_{\star0}
        \ar@{^{(}->}[r]
        \ar@{^{(}->}[d]&
        \W_0(\N)
        \ar@{^{(}->}[d]\\
        \N^+_\star
        \ar@{^{(}->}[r]&
        \W(\N)
}
\label{Wstar.states.weights.comm}
\end{equation}
commutes. For $\psi\in\W(\N)$,
\begin{equation}
        \supp(\psi)=\II-\sup\{P\in\Proj(\N)\mid\psi(P)=0\}.
\end{equation}
For $\omega,\phi\in\N_\star^+$ we will write $\rpktarget{ll}\omega\ll\phi$ if{}f $\supp(\omega)\leq\supp(\phi)$.\footnote{If $\N=\BH$ and $\omega=\tr(\rho_\omega\cdot)$ for $\rho_\omega\in\schatten_1(\H)^+$, then $\supp(\omega)=\ran(\rho_\omega)$, so for any $\phi=\tr(\rho_\phi\cdot)$ with $\rho_\phi\in\schatten_1(\H)^+$ one has $\omega\ll\phi$ if{}f $\ran(\rho_\omega)\subseteq\ran(\rho_\phi)$.} An element $\omega\in\N^{\banach+}$ is faithful if{}f $\supp(\omega)=\II$. If $\phi$ is a normal weight on a $W^*$-algebra $\N$ (which includes $\omega\in\N^+_\star$ as a special case), then the restriction of $\phi$ to a \df{reduced} $W^*$-algebra,
\begin{equation}
        \N_{\supp(\phi)}:=\{x\in\N\mid\supp(\phi)x=x=x\,\supp(\phi)\}=\bigcup_{x\in\N}\{\supp(\phi)x\,\supp(\phi)\},
\end{equation}
is a faithful normal weight (respectively, an element of $(\N_{\supp(\phi)})^+_{\star0}$). If $\phi$ is semi-finite, then $\phi|_{\N\supp(\phi)}\in\W_0(\N_{\supp(\phi)})$. Hence, given $\psi\in\W(\N)$ and $P\in\Proj(\N)$, $P=\supp(\psi)$ if{}f $\psi|_{\N_P}\in\W_0(\N_P)$ and $\psi(P)=\psi(PxP)\;\forall x\in\N^+$. In particular, for $\omega,\phi\in\N^+_\star$ and $\omega\ll\phi$, we have $\omega|_{\N_{\supp(\phi)}}\in\W_0(\N_{\supp(\phi)})$.

A \df{representation} of a $C^*$-algebra $\C$ is defined as a pair $(\H,\pi)$ of a Hilbert space $\H$ and a $*$-homomorphism $\rpktarget{pi}\pi:\C\ra\BH$. A representation $\pi:\C\ra\BH$ is called: \df{nondegenerate} if{}f $\{\pi(x)\xi\mid (x,\xi)\in\C\times\H\}$ is dense in $\H$; \df{normal} if{}f it is continuous with respect to the weak-$\star$ topologies of $\C$ and $\BH$; \df{faithful} if{}f $\ker(\pi)=\{0\}$. An element $\xi\in\H$ is called \df{cyclic} for a $C^*$-algebra $\C\subseteq\BH$ if{}f $\C\xi:=\bigcup_{x\in\C}\{x\xi\}$ is norm dense in $\BH$. A representation $\pi:\C\ra\BH$ of a $C^*$-algebra $\C$ is called \df{cyclic} if{}f there exists $\Omega\in\H$ that is cyclic for $\pi(\C)$. According to the Gel'fand--Na\u{\i}mark--Segal theorem \cite{Gelfand:Naimark:1943,Segal:1947:irreducible} for every pair $(\C,\omega)$ of a $C^*$-algebra $\C$ and $\omega\in\C^{\banach+}$ there exists a triple $(\H_\omega,\pi_\omega,\Omega_\omega)$ of a Hilbert space $\H_\omega$ and a cyclic representation $\pi_\omega:\C\ra\BH$ with a cyclic vector $\Omega_\omega\in\H_\omega$\rpktarget{h.omega}, and this triple is unique up to unitary equivalence. It is constructed as follows. For a $C^*$-algebra $\C$ and $\omega\in\C^{\banach+}$, one defines the scalar form $\rpktarget{scal.omega}\s{\cdot,\cdot}_\omega$ on $\C$,
\begin{equation}
        \s{x,y}_\omega := \omega(x^*y)\;\;\forall x,y\in\C,
\end{equation}
and the \df{Gel'fand ideal} 
\begin{equation}
        \I_\omega:=\{x\in\C\mid\omega(x^*x)=0\}=\{x\in\C\mid\omega(x^*y)=0\;\forall y\in\C\},
\end{equation}
which is a left ideal of $\C$, closed in the norm topology (it is also closed in the weak-$\star$ topology if $\omega\in\C^{\banach+}_\star$). The form $\s{\cdot,\cdot}_\omega$ is hermitean on $\C$ and it becomes a scalar product $\s{\cdot,\cdot}_\omega$ on $\C/\I_\omega$. The Hilbert space $\H_\omega$ is obtained by the completion of $\C/\I_\omega$ in the topology of norm generated by $\s{\cdot,\cdot}_\omega$. Consider the morphisms\rpktarget{rep.omega}
\begin{align}
        [\cdot]_\omega:\C\ni x&\longmapsto [x]_\omega\in\C/\I_\omega,\\
        \pi_\omega(y):[y]_\omega&\longmapsto[xy]_\omega.
\end{align}
The element $\omega\in\C^{\banach+}$ is uniquely represented in terms of $\H_\omega$ by the vector $[\II]_\omega=:\Omega_\omega\in\H_\omega$, which is cyclic for $\pi_\omega(\C)$ and satisfies $\n{\Omega_\omega}=\n{\omega}$. Hence
\begin{equation}
        \omega(x)=\s{\Omega_\omega,\pi_\omega(x)\Omega_\omega}_\omega
        \;\;\forall x\in\C,
\end{equation}
An analogue of this theorem for weights follows the similar construction, but lacks cyclicity. If $\N$ is a $W^*$-algebra, and $\omega$ is a weight on $\N$, then there exists the Hilbert space $\H_\omega$, defined as the completion of $\nnn_\omega/\ker(\omega)$ in the topology of a norm generated by the scalar product $\s{\cdot,\cdot}_{\omega}:\nnn_\omega\times\nnn_\omega\ni(x,y)\mapsto\omega(x^*y)\in\CC$,
\begin{equation}\rpktarget{h.omega.zwei}
        \H_\omega:=\overline{\nnn_\omega/\ker(\omega)}=\overline{\{x\in\N\mid\omega(x^*x)<\infty\}/\{x\in\N\mid\omega(x^*x)=0\}}=\overline{\nnn_\omega/\I_\omega},
\end{equation}
and there exist the maps\rpktarget{rep.omega.zwei}\rpktarget{pi.omega.zwei} 
\begin{align}
        [\cdot]_\omega:\nnn_\omega\ni x&\mapsto [x]_\omega\in\H_\omega,
        \label{GNS.class.weight}\\
        \pi_\omega:\N\ni x&\mapsto([y]_\omega\mapsto[xy]_\omega)\in\BBB(\H_\omega),
        \label{GNS.rep.weight}
\end{align}
such that $[\cdot]_\omega$ is linear, $\ran([\cdot]_\omega)$ is dense in $\H_\omega$, and $(\H_\omega,\pi_\omega)$ is a representation of $\N$. If $\omega\in\W(\N)$ then $(\H_\omega,\pi_\omega)$ is nondegenerate and normal. It is also faithful if $\omega\in\W_0(\N)$. The \df{commutant} of a subalgebra $\N$ of any algebra $\C$ is defined as 
\begin{equation}\rpktarget{comm}
        \N^\comm:=\{y\in\C\mid xy=yx\;\forall x\in\N\},
\end{equation}
while the \df{center} of $\N$ is defined as $\rpktarget{zentr}\zentr_\N:=\N\cap\N^\comm$. A unital $*$-subalgebra $\N$ of an algebra $\BH$ is called the \df{von Neumann algebra} \cite{vonNeumann:1930:algebra,Murray:vonNeumann:1936} if{}f $\N=\N^\comm{}^\comm$. An image $\pi(\N)$ of any representation $(\H,\pi)$ of a $W^*$-algebra $\N$ is a von Neumann algebra if{}f $\pi$ is normal and nondegenerate. 

A subspace $\D\subseteq\H$ of a complex Hilbert space $\H$ is called a \df{cone} if{}f $\lambda\xi\in\D$ $\forall\xi\in\D$ $\forall\lambda\geq0$. A cone $\D\subseteq\H$ is called \df{self-polar} if{}f
\begin{equation}
        \D=\{\zeta\in\H\mid\s{\xi,\zeta}_\H\geq0\;\forall\xi\in\D\}.
\end{equation}
Every self-polar cone $\D\subseteq\H$ is pointed ($\D\cap(-\D)=\{0\}$), spans linearly $\H$ ($\Span_\CC\D=\H$), and determines a unique conjugation\footnote{A linear operator $J:\dom(J)\ra\H$, where $\dom(J)\subseteq\H$, is called a \df{conjugation} if{}f it is antilinear, isometric, and involutive ($J^2=\II$).} $J$ in $\H$ such that $J\xi=\xi\;\forall\xi\in\H$ \cite{Haagerup:1973}, 
as well as a partial order on the set $\H^\sa:=\{\xi\in\H\mid J\xi=\xi\}$ given by,
\begin{equation}
        \xi\leq\zeta\;\iff\;\xi-\zeta\in\D\;\;\forall\xi,\zeta\in\H^\sa.
\end{equation}
If $\N$ is a $W^*$-algebra, $\H$ is a Hilbert space, $\stdcone\subseteq\H$ is a self-polar cone, $\pi$ is a nondegenerate faithful normal representation of $\N$ on $\H$, and $J$ is conjugation on $\H$, then the quadruple $(\H,\pi,J,\stdcone)$ is called \df{standard representation} of $\N$ and $(\H,\pi(\N),J,\stdcone)$ is called \df{standard form} of $\N$ if{}f the conditions \cite{Haagerup:1975:standard:form}
\begin{equation}
	J\pi(\N)J=\pi(\N)^\comm,\;\;
	\xi\in\stdcone\limp J\xi=\xi,\;\;
	\pi(x)J\pi(x)J\stdcone\subseteq\stdcone,\;\;
	\pi(x)\in\zentr_{\pi(\N)}\limp J\pi(x)J=\pi(x)^*.
\end{equation}
hold. For any standard representation
\begin{equation}
        \forall\phi\in\N_\star^+\;\exists !\xi_\pi(\phi)\in\stdcone\;\forall x\in\N\;\;\phi(x)=\s{\xi_\pi(\phi),\pi(x)\xi_\pi(\phi)}_\H
\label{std.vector.representative}
\end{equation}
holds. The map $\xi_\pi:\N^+_\star\ra\stdcone$ is order preserving. If $\N=\BBB(\K)$ for some Hilbert space $\K$, then the standard representation Hilbert space is given by the space $\K\otimes\K^\banach\iso\schatten_2(\K):=\{x\in\BBB(\K)\mid\sqrt{\tr_\K(x^*x)}<\infty\}$ of Hilbert--Schmidt operators equipped with an inner product $\s{x,y}:=\tr_\K(x^*y)$, where $(\cdot)^\banach$ denotes Banach dual space. In this case, the map \eqref{std.vector.representative} reads $\xi_\pi:\schatten_1(\H)^+\ni\rho\mapsto\rho^{1/2}\in\schatten_2(\H)^+$.
\subsection{Relative modular theory}
For a given $W^*$-algebra $\N$, $\phi\in\W(\N)$, and $\omega\in\W_0(\N)$ the map
\begin{equation}
        R_{\phi,\omega}:[x]_\omega\mapsto[x^*]_\phi\;\;\forall x\in\nnn_\omega\cap\nnn_\phi^*
        \label{relative.modular.weights}
\end{equation}
is a densely defined, closable antilinear operator. Its closure admits a unique polar decomposition
\begin{equation}
        \overline{R}_{\phi,\omega}=J_{\phi,\omega}\Delta^{1/2}_{\phi,\omega},
\end{equation}
where $\rpktarget{JREL}J_{\phi,\omega}$ is a conjugation operator, called \df{relative modular conjugation}, while $\rpktarget{DELTAREL}\Delta_{\phi,\omega}$ is a positive self-adjoint operator on $\dom(\Delta_{\phi,\omega})\subseteq\H_\omega$ with $\supp(\Delta_{\phi,\omega})=\supp(\phi)\H_\omega$, called a \df{relative modular operator} \cite{Araki:1973:relative:hamiltonian,Connes:1974,Digernes:1975}. The relative modular operators allow to define a one-parameter family of partial isometries in $\supp(\phi)\N$, called \df{Connes' cocycle} \cite{Connes:1973:classification},
\begin{equation}
\RR\ni t\mapsto\Connes{\phi}{\omega}{t}:=\Delta^{\ii t}_{\phi,\psi}\Delta^{-\ii t}_{\omega,\psi}=\Delta^{\ii t}_{\phi,\omega}\Delta^{-\ii t}_{\omega,\omega}\in\supp(\phi)\N,
\label{Connes.cocycle.def}
\end{equation}
where $\psi\in\W_0(\N)$ is arbitrary, so it can be set equal to $\omega$.  As shown by Araki and Masuda \cite{Araki:Masuda:1982} (see also \cite{Masuda:1984}), the definition of $\Delta_{\phi,\omega}$ and $\Connes{\phi}{\omega}{t}$ can be further extended to the case when $\phi,\omega\in\W(\N)$, by means of a densely defined closable antilinear operator
\begin{equation}
        R_{\phi,\omega}:[x]_\omega+(\II-\supp(\overline{[\nnn_\phi]_\omega}))\zeta\mapsto\supp(\omega)[x^*]_\phi\;\;\forall x\in\nnn_\omega\cap\nnn_\phi^*\;\forall\zeta\in\H,
\label{relative.modular.for.normal.weights}
\end{equation}
where $(\H,\pi,J,\stdcone)$ is a standard representation of a $W^*$-algebra $\N$, and $\H_\phi\subseteq\H\supseteq\H_\omega$. For $\phi,\omega\in\N_\star^+$ this becomes a closable antilinear operator \cite{Araki:1977:relative:entropy:II,Kosaki:1980:PhD}
\begin{equation}
        R_{\phi,\omega}:x\xi_\pi(\omega)+\zeta\mapsto\supp(\omega)x^*\xi_\pi(\phi)\;\;\forall x\in\pi(\N)\;\forall\zeta\in(\pi(\N)\xi_\pi(\omega))^\bot,
        \label{relative.modular.for.normal.states}
\end{equation}
acting on a dense domain $(\pi(\N)\xi_\pi(\omega))\cup(\pi(\N)\xi_\pi(\omega))^\bot\subseteq\H$, where $(\pi(\N)\xi_\pi(\omega))^\bot$ denotes a complement of the closure in $\H$ of the linear span of the action $\pi(\N)$ on $\xi_\pi(\omega)$. In both cases, the relative modular operator is determined by the polar decomposition of the closure $\overline{R}_{\phi,\omega}$ of $R_{\phi,\omega}$,\rpktarget{DELTA.REL.ZWEI}
\begin{equation}
        \Delta_{\phi,\omega}:=R^*_{\phi,\omega}\overline{R}_{\phi,\omega}.
\label{RR.Delta.relative}
\end{equation}
If \eqref{relative.modular.for.normal.weights} or \eqref{relative.modular.for.normal.states} is used instead of \eqref{relative.modular.weights}, then the formula \eqref{Connes.cocycle.def} has to be replaced by\rpktarget{CONNES.COC.ZWEI}
\begin{equation}
        \RR\ni t\mapsto\Connes{\phi}{\omega}{t}\supp(\overline{[\nnn_\phi]_\psi}):=\Delta^{\ii t}_{\phi,\psi}\Delta^{-\ii t}_{\omega,\psi},
\label{Connes.for.ns.weights}
\end{equation}
and $\Connes{\phi}{\omega}{t}$ is a partial isometry in $\supp(\phi)\N\supp(\omega)$ whenever $[\supp(\phi),\supp(\omega)]=0$.
\subsection{Quantum distances}
Given any set $X$, a \df{distance} is defined as a map $D:X\times X\ra[0,\infty]$ such that $D(x,y)=0$ $\iff$ $x=y$. A \df{relative entropy} is a map $\entropy:X\times X\ra[-\infty,0]$ such that $(-\entropy)$ is a distance. A distance is called: \df{bounded} if{}f $\ran(D)=\RR^+$; \df{symmetric} if{}f $D(x,y)=D(y,x)$; \df{metrical} \cite{Frechet:1906} if{}f is it bounded, symmetric and satisfies \df{triangle inequality}
\begin{equation}
	D(x,y)\leq D(x,y)+D(y,z)\;\;\forall x,y,z\in X.
\end{equation}
We will use the symbol $d$ instead of $D$ to denote metrical distances. A distance on a statistical model will be called a \df{statistical distance}, while a distance on a quantum model will be called a \df{quantum distance}. A term \df{information distance} with be used to refer to any of them. Let $\Pow(X)$ denote a powerset of $X$. If $K\subseteq X$ is such that the map
\begin{equation}
	\PPP^D_K:X\ni\phi\mapsto\arginff{\omega\in K}{D(\omega,\phi)}\in\Pow(K)
\label{entropic.preprojection.map}
\end{equation}
is a singleton (one-element set, $\{*\}$) for all $\phi\in X$, then we call \eqref{entropic.preprojection.map} an \df{entropic projection}\footnote{So, by definition, every entropic projection $\PPP^D_K$ is a unique maximiser of a relative entropy $-D$.}.

For any $W^*$-algebra $\N$, and $\phi,\omega\in\N^+_\star$ the following formulas define quantum distances on $\N^+_\star$,
\begin{equation}
	D_0(\omega,\phi):=
	        \left\{
                \begin{array}{ll}
                        (\omega-\phi)(\II)-\s{\xi_\pi(\phi),\log(\Delta_{\omega,\phi})\xi_\pi(\phi)}_\H&:\omega\ll\phi\\
                        +\infty&:\mbox{otherwise},
                \end{array}
        \right.	
\end{equation}
\begin{equation}
	D_1(\omega,\phi):=
	        \left\{
                \begin{array}{ll}
                        (\phi-\omega)(\II)-\s{\xi_\pi(\omega),\log(\Delta_{\omega,\phi})\xi_\pi(\omega)}_\H&:\omega\ll\phi\\
                        +\infty&:\mbox{otherwise},
                \end{array}
        \right.	
\end{equation}
\begin{equation}
	D_{1/2}(\psi,\phi):=
	        \left\{
                \begin{array}{ll}
                        2(\phi+\omega)(\II)-4\s{\xi_\pi(\phi),\Delta^{1/2}_{\omega,\phi}\xi_\pi(\phi)}_\H&:\omega\ll\phi\\
                        +\infty&:\mbox{otherwise}.
                \end{array}
        \right.
\label{D.one.half.distance}
\end{equation}
Hence,
\begin{equation}
	\phi\ll\omega\ll\phi\limp D_0(\omega,\phi)=D_1(\phi,\omega).
\end{equation}
All above examples are special cases of a family of quantum $\gamma$-distances $D_\gamma$ with $\gamma\in[0,1]$, see \cite{Hasegawa:1993,Ojima:2004,Jencova:2005,Kostecki:2011:OSID}. A special case of $D_1$ is the \df{Araki distance} \cite{Araki:1976:relham:relent,Araki:1976:relative:entropy:I,Araki:1977:relative:entropy:II} 
\begin{equation}
        D_1|_{\N^+_{\star1}}(\omega,\phi)=
        \left\{
                \begin{array}{ll}
                        -\s{\xi_\pi(\omega),\log(\Delta_{\phi,\omega})\xi_\pi(\omega)}_\H&:\omega\ll\phi\\
                        +\infty&:\mbox{otherwise}.
                \end{array}
        \right.
\label{Araki.distance}
\end{equation}
If $D_1|_{\N^+_{\star1}}(\omega,\phi)<\infty$, then \eqref{Araki.distance} takes the form \cite{Petz:1985:properties,Petz:1986:properties}
\begin{equation}
D_1|_{\N^+_{\star1}}(\omega,\phi)
=\left\{
\begin{array}{ll}
\ii\lim_{t\ra^+ 0}\frac{\omega}{t}\left(\Connes{\phi}{\omega}{t}-\II\right)
&:\;\omega\ll\phi\\
+\infty&:\;\mbox{otherwise}.
\end{array}
\right.
\label{Araki.Petz.distance}
\end{equation}
For a semi-finite $\N$, normal faithful semi-finite trace $\tau$ on $\N$ and $\rho_\phi$ and $\rho_\omega$ defined by $\phi(\cdot)=\tau(\rho_\phi\;\cdot)$ and $\omega(\cdot)=\tau(\rho_\omega\;\cdot)$, the Araki distance \eqref{Araki.distance} turns to the \df{Umegaki distance} \cite{Umegaki:1961,Umegaki:1962} (cf. also \cite{Araki:1976:relham:relent,Araki:1976:positive:cone})
\begin{equation}
        D_1|_{\N^+_{\star1}}(\omega,\phi)=\tau(\rho_\omega(\log\rho_\omega-\log\rho_\phi))=\tau\left(\rho_\omega^{1/2}(\log\Delta_{\omega,\phi})\rho_\omega^{1/2}\right)
\label{Umegaki.Araki.distance}
\end{equation}
if $\omega\ll\phi$, and $D_1|_{\N^+_{\star1}}(\omega,\phi)=+\infty$ otherwise. A special, but definitely most popular, case of \eqref{Umegaki.Araki.distance} is obtained for $\N=\BH$ and $\tau=\tr$. 

For a commutative $\N\iso L_\infty(\X,\mho(\X),\tmu)$, where $(\X,\mho(\X),\tmu)$ is any localisable measure space such that $\tmu_\phi\ll\tmu$ and $\tmu_\omega\ll\tmu$ where $\psi(x)=:\int_\X\tmu_\psi(\xx)x(\xx)$ $\forall x\in\N^+$ and $\psi\in\{\phi,\omega\}$, the Araki distance \eqref{Araki.distance} takes a form of the \df{Wald--Good--Kullback--Leibler distance} \cite{Wald:1947,Good:1950,Kullback:Leibler:1951},
\begin{equation}
        D_1|_{L_1(\X,\mho(\X),\tmu)^+_1}(\omega,\phi)=
        \left\{
                \begin{array}{ll}
                        \int\tmu_\omega\log\frac{\tmu_\omega}{\tmu_\phi}&:\tmu_\omega\ll\tmu_\phi\\
                        +\infty&:\mbox{otherwise}.
                \end{array}
        \right.
\label{WGKL.distance}
\end{equation}
Under these conditions, and assuming $\tmu_\phi\ll\tmu_\omega$, one has 
\begin{equation}
	D_1|_{L_1(\X,\mho(\X),\tmu)^+_1}(\omega,\phi)=D_0|_{L_1(\X,\mho(\X),\tmu)^+_1}(\phi,\omega).
\end{equation}

Among above distances, only $D_{1/2}$ is symmetric, and none of them is metrical. Metrical quantum distances will be discussed in Section \ref{transition.correlation.section}.
\section{L\"{u}ders' and quantum Jeffrey's rules as entropic projections\label{vN.Lu.section}}
\begin{definition}
If $\rho\in\schatten_1(\H)^+$, $I$ is a countable set, and $\{P_i\mid i\in I\}\subseteq\Proj(\BH)$ satisfies $\sum_{i\in I}P_i=\II\in\BH$ and $P_iP_j=\dirac_{ij}P_i$ $\forall i,j\in I$, then the \df{weak L\"{u}ders rule} is defined as a map \cite{Lueders:1951,Schwinger:1959}
\begin{equation}
	\schatten_1(\H)^+\ni\rho\mapsto\sum_{i\in I}P_i\rho P_i\in\schatten_1(\H)^+.
\label{weak.Lueders.rule}
\end{equation}
If $P\in\Proj(\BH)$, then the \df{strong L\"{u}ders rule} is defined as a map
\begin{equation}
	\schatten_1(\H)^+_1\ni\rho\mapsto\frac{P\rho P}{\tr(P\rho P)}\in\schatten_1(\H)^+_1,
\label{strong.Lueders.rule}
\end{equation}
with domain restricted by the condition $\tr(P\rho)\neq0$. The \df{semi-strong L\"{u}ders rule}, called sometimes a `partial collapse', is defined as a map
\begin{equation}
	\schatten_1(\H)^+_1\ni\rho\mapsto\frac{\sum_{j\in J} P_j\rho P_j}{\sum_{j\in J}\tr(P_j\rho P_j)}\in\schatten_1(\H)^+_1,
\label{semistrong.Lueders.rule}
\end{equation}
where $J$ is a subset of a countable set $I$ corresponding to an orthogonal decomposition $\sum_{i\in I}P_i=\II\in\BH$, and the domain in \eqref{semistrong.Lueders.rule} is restricted by a condition $\sum_{j\in J}\tr(P_j\rho)\neq0$.
\end{definition}

\begin{remark}
If $\rho$, $I$, and $\{P_i\mid i\in I\}$ are such as in \eqref{weak.Lueders.rule}, $J\subseteq I$, and the condition
\begin{equation}
	\tr(P_i\rho)\neq0\;\;\iff\;\;i\in J\;\;\;\forall i\in I
\label{strong.from.weak.condition}
\end{equation}
holds, then \eqref{weak.Lueders.rule} can be written as
\begin{equation}
	\schatten_1(\H)^+\ni\rho\mapsto\sum_{j\in J}\frac{P_j\rho P_j}{\tr(\rho P_j)}\tr(\rho P_j)\in\schatten_1(\H)^+.
\label{weak.goes.strong.equation}
\end{equation}
In particular, if $J=\{*\}$, then \eqref{weak.goes.strong.equation} turns to the direct extension of strong L\"{u}ders' rule \eqref{strong.Lueders.rule} to $\schatten_1(\H)^+$, which coincides with \eqref{strong.Lueders.rule} on $\schatten_1(\H)^+_1$.
\end{remark}

\begin{remark}
In the case when $\rho^2=\rho$, then $\exists\xi\in\H$ such that $\rho=P_{\Span(\xi)}$ (i.e., $\rho$ is a projector onto a closed one-dimensional linear subspace of $\H$ spanned by $\xi$), and the strong L\"{u}ders rule takes the form of the \df{strong von Neumann rule} \cite{vonNeumann:1932:grundlagen}, called also a `state vector reduction',
\begin{equation}
	\H\ni\xi\mapsto\frac{P\xi}{\s{\xi,P\xi}^{\frac{1}{2}}}\in\H.
\label{strong.vN.rule}
\end{equation}
\end{remark}

\begin{definition}\label{quantum.Jeffrey.rule.definition}
We define \df{quantum Jeffrey's rule} as a map
\begin{equation}
	\schatten_1(\H)^+_1\ni\rho\mapsto\rho_{\mathrm{new}}:=\sum_{i=1}^n\frac{P_i\rho P_i}{\tr(\rho P_i)}\lambda_i\in\schatten_1(\H)^+_1,
\label{quantum.Jeffrey.rule}
\end{equation}
where $n\in\NN$, $\{P_i\}_{i=1}^n\subseteq\Proj(\BH)$, $\sum_{i=1}^nP_i=\II\in\BH$, and $\tr(\rho P_i)\neq0$, $P_iP_j=\dirac_{ij}P_i$, $\lambda_i=\tr\left(\rho_{\mathrm{new}}P_i\right)$ $\forall i,j\in\{1,\ldots,n\}$.
\end{definition}

\begin{remark}\label{strongLuedersfromqJr.remark}
In the case when an orthogonal decomposition of $\II\in\BH$ is given by the set $\{P, \II-P\}$, and $\tr(\rho_{\mathrm{new}}(\II-P))=0$, then \eqref{quantum.Jeffrey.rule} reduces to \eqref{strong.Lueders.rule}. The direct analogy between this property and the conditions under which Jeffrey's rule \eqref{Jeffrey.rule} reduces to the Bayes--Laplace rule \eqref{Bayes.rule.eq}, together with the analogy between \eqref{quantum.Jeffrey.rule} and \eqref{Jeffrey.rule}, justify the name we gave to \eqref{quantum.Jeffrey.rule}.
\end{remark}

\begin{proposition}
Definition \eqref{quantum.Jeffrey.rule} of $\rho_{\mathrm{new}}$ is equivalent to a condition
\begin{equation}
	\frac{\tr(\rho_{\mathrm{new}}P_iP)}{\tr(\rho_{\mathrm{new}}P_i)}=\frac{\tr(\rho P_iP)}{\tr(\rho P_i)}\;\;\;\forall P\in\Proj(\BH)\;\mbox{such that}\;[P,P_i]=0\;\;\forall i\in\{1,\ldots,n\},
\label{equivalent.quantum.Jeffrey.condition}
\end{equation}
with $[\rho_{\mathrm{new}},P_i]=0$ $\forall i\in\{1,\ldots,n\}$.
\end{proposition}

\begin{proof}
An implication from \eqref{quantum.Jeffrey.rule} to \eqref{equivalent.quantum.Jeffrey.condition} is easy. In the opposite direction, let us consider an arbitrary $x\in\BH^\sa$ such that $[x,P_i]=0$ $\forall i\in\{1,\ldots,n\}$. For an arbitrary countable set $I$ the condition $[P_i,\rho_{\mathrm{new}}]=0$ $\forall i\in I$ is equivalent to $\rho_{\mathrm{new}}=\sum_{i\in I}P_i\rho_{\mathrm{new}}P_i$ and to $[\rho_{\mathrm{new}},x]=0$, where $x=\sum_{i\in I}\lambda_iP_i\in\BH^\sa$ with arbitrary $\{\lambda_i\in\RR\mid i\in I\}$. See \cite{Herbut:1969} for a clear discussion of these conditions and their equivalence. Hence,
\begin{align}
\tr(\rho_{\mathrm{new}}x)&=\sum_{i=1}^n\tr(P_i\rho_{\mathrm{new}}P_ix)=\sum_{i=1}^n\tr(\rho_{\mathrm{new}}P_ixP_i)=\sum_{i=1}^n\tr(\rho_{\mathrm{new}}P_i)\frac{\tr(\rho P_ixP_i)}{\tr(\rho P_i)}\nonumber\\&=\sum_{i=1}^n\tr(\rho_{\mathrm{new}}P_i)\frac{\tr(P_i\rho P_ix)}{\tr(\rho P_i)}.
\end{align}
This gives
\begin{equation}
\rho_{\mathrm{new}}=\sum_{i=1}^n\tr(\rho_{\mathrm{new}}P_i)\frac{P_i\rho P_i}{\tr(\rho P_i)}.
\end{equation}
\end{proof}

\begin{remark}
In what follows we will prove that the weak L\"{u}ders rule is a special case of an entropic projection, determined by constrained minimisation with a $D_0(\phi,\psi)$ distance (Theorems \ref{vonNeumann.Lueders.from.entropy.theorem} and \ref{gen.vonNeumann.Lueders.from.entropy.theorem}). Next, we will use different constraints to derive quantum Jeffrey's rule (Theorem \ref{weighted.vNLu.theorem}). Finally, we will show that the strong L\"{u}ders rule arises as a weakly continuous limit of quantum Jeffrey's rule (Remark \ref{remark.strong.vNLu}), and can be also obtained by a constrained minimisation of a \textit{regularised} modification of $D_0(\phi,\psi)$ distance (Theorem \ref{regularised.MRE.vN.theorem}). The necessity of recourse to limit or regularisation indicates that, as opposed to weak L\"{u}ders' rule, strong L\"{u}ders' rule is not directly derivable from minimisation of $D_0$ distance. However, in Section \ref{transition.correlation.section} we will show that in some special cases (which usually turn it to the strong von Neumann rule) it can be derived from minimisation of $D_{1/2}$ distance.
\end{remark}

\begin{remark}
According to de Muynck \cite{deMuynck:2002} (see also \cite{Fok:1932:nachala,Margenau:1936,Kemble:1937,Margenau:1963}), the strong von Neumann and strong L\"{u}ders rules should be viewed as procedures of updating of quantum state, yet not in the \textit{predictive} (inferential) sense, but in the \textit{preparative} (calibrating) sense. If one extends this interpretation to the weak L\"{u}ders rule, and assumes that the `preparative' use of entropic projection should be implemented by constrained minimisation of information distance in its \textit{second} argument, then our result can be interpreted by saying that the weak L\"{u}ders rule of \textit{quantum state preparation} can be derived from the constrained minimisation of the Araki distance $\dono(\omega,\phi)$. (We do not consider this interpretation as necessary.)
\end{remark}

\begin{remark}
In what follows $\M$ will be an arbitrary set, $\Q\subseteq\M$ its arbitrary subset, and $D$ will be an arbitrary (not necessarily bounded) distance on $\M$. If $\arginff{\phi\in\Q}{D(\phi,\psi)}$ consists of a single element, then we will denote it by $\PPP^D_{\Q}(\psi)$.
\end{remark}

\begin{definition}
Let $\psi\in\M$ and $\Q\subseteq\M$. If
\begin{equation}
        \exists\rho\in\Q\;
        \forall\phi\in\Q\;\;
        D(\phi,\rho)+D(\rho,\psi)=D(\phi,\psi)
\label{triangle.equality}
\end{equation}
holds, then we will say that $\Q$ satisfies \df{triangle equality} for $\psi$ at $\rho\in\Q$ with respect to $D$.
\end{definition}

\begin{definition}
Let $\Q_1,\Q_2\subseteq\M$. If
\begin{equation}
        \forall\psi\in\Q_1\;\;\;\arginff{\phi\in\Q_2}{D(\phi,\psi)}\iso\{*\}\;\mbox{ and }\;\PPP^D_{\Q_2}(\psi)\in\Q_1,
\label{subprojection.property}
\end{equation}
then we will say that the pair $(\Q_1,\Q_2)$ satisfies \df{subprojection property} with respect to $D$.
\end{definition}



\begin{lemma}
If $\Q\subseteq\M$ satisfies triangle equality for $\psi$ at $\rho\in\Q$ with respect to $D$, then
\begin{equation}
        \rho=\arginff{\phi\in\Q}{D(\phi,\psi)}
\end{equation}
\end{lemma}

\begin{proof}
From $D(\phi,\rho)\geq0$ and $D(\phi,\rho)=0\iff\phi=\rho$ it follows that $\rho=\arginff{\phi\in\Q}{D(\phi,\rho)}$. From 
\begin{equation}
        \arginff{\phi\in\Q}{D(\phi,\rho)}=
        \arginff{\phi\in\Q}{D(\rho,\psi)+D(\phi,\rho)}
\end{equation}
it follows that $\arginff{\phi\in\Q}{D(\phi,\psi)}$ exists, is unique, and is equal to $\rho$.
\end{proof}

\begin{lemma}\label{lemma.subsubprojection}
Let $(\Q_1,\Q_2)$ satisfy the subprojection property with respect to $D$, and let $\Q_1,\Q_2\in\M$ satisfy triangle equality for $\psi\in\M$ at $\PPP^D_{\Q_2}(\psi)$ with respect to $D$. Then $\Q_1\cap\Q_2$ satisfies triangle equality for $\psi\in\M$ at $\PPP^D_{\Q_1\cap\Q_2}(\psi)$ with respect to $D$ and
\begin{equation}
        \arginff{\phi\in{\Q_1}\cap{\Q_2}}{D(\phi,\psi)}=
        \arginff{%
        \phi_2\in{\Q_2}}%
        {D\left(\phi_2,\arginff{\phi_1\in{\Q_1}}{D(\phi_1,\psi)}\right)}.
\end{equation}
\end{lemma}

\begin{proof}
Triangle equalities in this case read
\begin{align}
        \exists\rho_1\in\Q_1\;
        \forall\phi_1\in\Q_1\;\;
        D(\phi_1,\rho_1)+D(\rho_1,\psi)
        &
        =D(\phi_1,\psi),\\
        \exists\rho_2\in\Q_2\;
        \forall\phi_2\in\Q_2\;\;
        D(\phi_2,\rho_2)+D(\rho_2,\rho_1)
        &
        =D(\phi_2,\rho_1).
\end{align}
Now, let $\rho_2\in\Q_1\cap\Q_2$. This gives
\begin{equation}
        D(\phi,\rho_1)+D(\rho_2,\rho_1)+D(\rho_1,\psi)=
        D(\rho_2,\psi)+D_1(\phi,\rho_2)+D(\rho_2,\rho_1).
\end{equation}
For $\phi\in\Q_1\cap\Q_2$ we have
\begin{equation}
        D(\phi,\rho_1)+D(\rho_1,\psi)=D(\phi,\psi).
\end{equation}
This gives
\begin{equation}
        D(\phi,\rho_2)+D(\rho_2,\psi)=D(\phi,\psi).
\end{equation}
\end{proof}

\begin{proposition}\label{lemma.n.triangle}
If $\Q_i$ satisfies triangle equality for every $i\in\{1,\ldots,n\}$, and $(\Q_i,\Q_j)$ satisfy subprojection property for every $i,j\in\{1,\ldots,n\}$, then
\begin{equation}
        \arginff{\phi\in\Q_1\cap\ldots\cap\Q_n}{D(\phi,\psi)}=\rho_n,
\end{equation}
where $\rho_k=\arginff{\phi\in\Q_k}{D(\phi,\rho_{k-1})}$, $\rho_0=\psi$, and $\Q_n$ satisfies triangle equality at $\rho_n$.
\end{proposition}

\begin{proof}
We will prove this lemma by mathematical induction. Let us assume that it holds for some $k\in\NN$. Then
\begin{equation}
        \rho_k=\arginff{\phi\in\Q_1\cap\ldots\cap\Q_k}{D(\phi,\psi)}
\end{equation}
and $\Q_k$ satisfies triangle equality with $\rho_k$. Let $\psi_k\in\Q_1\cap\ldots\cap\Q_k$, and consider
\begin{equation}
        \PPP^D_{\Q_{k+1}}(\psi_k):=\arginff{\phi\in\Q_{k+1}}{D(\phi,\psi_k)}.
\end{equation}
Then $\psi_k\in\Q_i$ for every $i\leq k$, and from the subprojection property for $(\Q_i,\Q_{k+1})$ it follows that $\PPP^D_{\Q_{k+1}}(\psi_k)\in\Q_i$, so $\PPP^D_{\Q_{k+1}}(\psi_k)\in\Q_1\cap\ldots\cap\Q_k$. Hence, subprojection property holds for $\Q_1\cap\ldots\cap\Q_k$ and $\Q_{k+1}$. Lemma \ref{lemma.subsubprojection} applied to $\Q_1\cap\ldots\cap\Q_k$ and $\Q_{k+1}$ gives
\begin{equation}
\PPP^D_{\Q_{k+1}}(\rho_k)=\arginff{\phi\in(\Q_1\cap\ldots\cap\Q_n)\cap\Q_{k+1}}{D(\phi,\psi)}.
\end{equation}
This lemma holds for $k=2$ by subprojection property of $\Q_1$ and $\Q_2$.
\end{proof}

\begin{remark}
In what follows, we will assume that $\N=\BH$ and $\phi,\psi,\omega,\rho\in\M(\N)=\BH_{\star1}^+\iso{\schatten_1(\H)}^+_1$ for some Hilbert space $\H$ of arbitrary dimension. The value of $n\in\NN$ will be kept arbitrary but fixed.
\end{remark}



\begin{lemma}\label{tracelog.triangle.lemma}
If $\tr(\psi(\log\rho-\log\psi)^2)<\infty$, then the triangle equality at $\rho$ for $\psi$ with respect to $D_0$ is equivalent with
\begin{equation}
        \forall\psi\in\M(\N)\;
        \exists\rho\in\Q\;
        \forall\phi\in\Q\;\;
        \tr(\psi(\log\rho-\log\phi))=\tr(\rho(\log\rho-\log\phi)).
\end{equation}
\end{lemma}

\begin{proof}
If these distance functionals are finite, then
\begin{align}
        \tr(\psi(\log\psi-\log\phi))&=
        \tr(\rho(\log\rho-\log\phi))+\tr(\psi(\log\psi-\log\rho)),\\
        \tr(\psi(\log\rho-\log\phi))&=\tr(\rho(\log\rho-\log\phi)).
\end{align}
\end{proof}

\begin{proposition}\label{small.theorem}
Let $\Q=\{\omega\in\M(\N)\mid [\omega,P]=0\}$, and let $\rho=P\psi P+(\II-P)\psi(\II-P)$, where $P\in\Proj(\N)$. Then
\begin{equation}
        \forall\phi\in\Q\;\;
        D_0(\phi,\psi)=D_0(\phi,\rho)+D_0(\rho,\psi).
\label{equation.of.small.theorem}
\end{equation}
\end{proposition}

\begin{proof}
The operators $\phi,\rho$ are block-diagonal matrices, so any functions of $\rho$ and $\phi$ are also block-diagonal. Thus, $[\log\rho-\log\phi,P]=0$ and $[\log\rho-\log\phi,\II-P]=0$. It follows
\begin{equation}
        (\log\rho-\log\phi)=
        (\log\rho-\log\phi)(P^2+(\II-P)^2)=
        P(\log\rho-\log\phi)P+(\II-P)(\log\rho-\log\phi)(\II-P).
\end{equation}
Substituting $P\rho P=P(P\psi P+(\II-P)\psi(\II-P))P=P\psi P$ 
we get 
\begin{equation}
        \tr(\psi P(\log\rho-\log\phi)P)=\tr(\rho P(\log\rho-\log\phi)P).
\end{equation}
Similarly, we obtain    
\begin{equation}
        \tr(\psi(\II-P)(\log\rho-\log\phi)(\II-P))=\tr(\rho(\II-P)(\log\rho-\log\phi)(\II-P)),
\end{equation}
and it follows that
\begin{equation}
        \tr(\psi(\log\rho-\log\phi))=\tr(\rho(\log\rho-\log\phi)).
\end{equation}
By Lemma \ref{tracelog.triangle.lemma}, this gives \eqref{equation.of.small.theorem}, but it remains to check whether the assumption of this lemma is satisfied. We have
\begin{equation}
 \n{\rho^{1/2}\psi^{-1/2}\psi^{1/2}}_{\schatten_2(\H)}=\tr(\rho)<\infty.
\label{rho.one.half.less.infty}
\end{equation}
Now we want to show that $\n{\rho^{-1/2}\psi}^2_{\schatten_2(\H)}<\infty$. Let us denote in matrix form
\begin{equation}
        \psi=
        \left(
        \begin{array}{cc}
                \psi_{11}&\psi_{12}\\
                \psi_{21}&\psi_{22}
        \end{array}
        \right),\;\;\;
        \rho=
        \left(
        \begin{array}{cc}
                \psi_{11}&0\\
                0&\psi_{22}
        \end{array}
        \right).
\end{equation}
From $\psi\geq0$ it follows that $\psi_{22}\xi=0$ $\limp$ $\psi_{12}\xi=0$. Moreover,
\begin{equation}
	\psi_{11}-\left(\psi_{22}^{-1/2}\psi_{21}\right)^*\left(\psi_{22}^{-1/2}\psi_{21}\right)=\psi_{11}-\psi_{12}\psi^{-1}_{22}\psi_{21}.
\end{equation}
So, for $\xi\in\dom(\psi_{22}^{-1}\psi_{21})$, the corresponding forms satisfy
\begin{equation}
	\s{\xi,\psi_{11}\xi}_\H-\n{\psi_{22}^{-1/2}\psi_{21}\xi}_\H^2
	=\s{\left(
        \begin{array}{c}
        \II\\
        -\psi^{-1}_{22}\psi_{21}
        \end{array}
        \right)\xi,
			\left(
				\begin{array}{cc}
					\psi_{11}&\psi_{12}\\
					\psi_{21}&\psi_{22}
				\end{array}
			\right)
			\left(
        \begin{array}{c}
        \II\\
        -\psi^{-1}_{22}\psi_{21}
        \end{array}
        \right)
			\xi
			}_\H\geq0.
\end{equation}
Hence, as operators,
\begin{align}
				\psi_{11}-\psi_{12}\psi_{22}^{-1}\psi_{21}&\geq0,\\
        \psi_{12}\psi_{22}^{-1}\psi_{21}
        &\leq\psi_{11},\\
        \tr(\psi_{12}\psi_{22}\psi_{21}^{-1})
        &\leq\tr(\psi_{11}).
\end{align}
Hence
\begin{align}
        \n{(\rho^{-\frac{1}{2}}\psi^{\frac{1}{2}})\psi^{\frac{1}{2}}}^2_{\schatten_2(\H)}
        &=\tr(\psi\rho^{-1}\psi)
				=\tr
        \left(\begin{array}{cc}
                \psi_{12}\psi_{22}^{-1}\psi_{21}+\psi_{11}&
                \psi_{12}+\psi_{12}\\
                \psi_{21}+\psi_{21}&
                \psi_{21}\psi_{11}^{-1}\psi_{12}+\psi_{22}              
        \end{array}
        \right)\nonumber\\
				&\leq2\tr\left(\begin{array}{cc}
								\psi_{11}&
								\psi_{12}\\
								\psi_{21}&
								\psi_{22}
				\end{array}\right)
				=2\tr(\psi)
				<\infty.\label{rho.minus.one.half.less.infty}
\end{align}
Using
\begin{equation}
        \exists\lambda\in\RR^+\;\;\forall\gamma\in\RR^+\;\;\ab{\log(\gamma)}\leq\lambda\max(\gamma^{1/2},\gamma^{-1/2}),
\end{equation}
together with \eqref{rho.one.half.less.infty} and \eqref{rho.minus.one.half.less.infty}, we obtain
\begin{equation}
        \tr(\psi(\log\rho-\log\psi)^2)<\infty,
\end{equation}
which follows from
\begin{align}
	\tr\left(\psi(\log\rho-\log\psi)^2\right)&=\n{(\log\rho-\log\psi)\psi^{1/2}}^2_{\schatten_2(\H)}\nonumber\\
	&\leq\left(\n{\log(\rho)\psi^{1/2}}_{\schatten_2(\H)}+\n{\log(\psi)\psi^{1/2}}_{\schatten_2(\H)}\right)^2<\infty.
\end{align}
\end{proof}

\begin{lemma}\label{lemma.PiPj}
Given $P_i,P_j\in\Proj(\N)$, let $\Q_k:=\{\omega\in\M(\N)\mid[P_k,\omega]=0\}$ for $k\in\{i,j\}$, and $[P_i,P_j]=0$. If $\psi\in\Q_j$ then 
\begin{equation}
        \arginff{\phi\in\Q_i}{D_0(\phi,\psi)}\in\Q_j.
\end{equation}
\end{lemma}

\begin{proof}
From Proposition \ref{small.theorem} it follows that 
\begin{equation}
        \arginff{\phi\in\Q_i}{D_0(\phi,\psi)}=
        P_i\psi P_i+(\II-P_i)\psi(\II-P_i).
\end{equation}
From $[\psi,P_j]=0$ we obtain
\begin{equation}
        [P_i\psi P_i+(\II-P_i)\psi(\II-P_i),P_j]=
        P_i[\psi,P_j]P_i+(\II-P_i)[\psi,P_j](\II-P_i).
\end{equation}
\end{proof}

\begin{lemma}
Let $\{P_i\}_{i=1}^n\subseteq\Proj(\N)$, $[P_i,P_j]=0$ $\forall i,j\in\{1,\ldots,n\}$, and $\Q:=\{\omega\in\M(\N)\mid [P_i,\omega]=0\;\forall i\in\{1,\ldots,n\}\}$. Then $\Q$ satisfies triangle equality and 
\begin{equation}
        \arginff{\phi\in\Q}{D_0(\phi,\psi)}=\rho_n,
\end{equation}
where $\rho_k=\arginff{\phi\in\Q_k}{D_0(\phi,\rho_{k-1})}$ for $k\in\{1,\ldots,n\}$ and $\rho_0=\psi$.
\end{lemma}
\begin{proof}
Follows directly from Proposition \ref{lemma.n.triangle}, Lemma \ref{lemma.PiPj} and Proposition \ref{small.theorem}.
\end{proof}

\begin{theorem}\label{vonNeumann.Lueders.from.entropy.theorem}
If $\{P_i\}_{i=1}^n\subseteq\Proj(\N)$ satisfies $P_iP_j=\dirac_{ij}P_i$ $\forall i,j\in\{1,\ldots,n\}$, $\sum_{i=1}^nP_i=\II$, and 
\begin{equation}
        \Q_{\mathrm{L}}:=\{\omega\in\M(\N)\mid [P_i,\omega]=0\;\forall i\in\{1,\ldots,n\}\},
\label{QL.constraint}
\end{equation}
then 
\begin{equation}
        \PPP^{D_0}_{\Q_{\mathrm{L}}}(\psi)\equiv\arginff{\phi\in\Q_{\mathrm{L}}}{D_0(\phi,\psi)}=\sum_{i=1}^nP_i\psi P_i.
\label{weak.vonNeumann.Lueders}
\end{equation}
\end{theorem}

\begin{proof}
By mathematical induction. Assume that 
\begin{equation}
        \rho_k=
        \sum_{i=1}^kP_i\psi P_i+
        \left(\II-\sum_{i=1}^kP_i\right)
        \psi
        \left(\II-\sum_{i=1}^kP_i\right).
\end{equation}
Then
\begin{equation}
        \rho_{k+1}=
        P_{k+1}\rho_kP_{k+1}+(\II-P_{k+1})\rho(\II-P_{k+1})=
        P_{k+1}\psi P_{k+1}+
        \sum_{i=1}^kP_i\psi P_i +
        \left(\II-\sum_{i=1}^{k+1}P_i\right)
        \psi
        \left(\II-\sum_{i=1}^{k+1}P_i\right),
\end{equation}
what follows from
\begin{equation}
        (\II-P_{k+1})\left(\II-\sum_{i=1}^kP_i\right)=
        \left(\II-\sum_{i=1}^{k+1}P_i\right).
\end{equation}
The first step of this induction is satisfied by Proposition \ref{small.theorem}.
\end{proof}

\begin{lemma}\label{commutation.as.expectation.lemma}
If $\{P_i\}_{i=1}^n\subseteq\Proj(\N)$ satisfies $P_iP_j=\dirac_{ij}P_i$ $\forall i,j\in\{1,\ldots,n\}$ and $\sum_{i=1}^nP_i=\II$, then the conditions
\begin{enumerate}
\item[a)] $\Q=\{\omega\in\M(\N)\mid [P_i,\omega]=0\;\forall i\in\{1,\ldots,n\}\}$
\item[b)] $\Q=\{\omega\in\M(\N)\mid \tr(\omega[f(\{P_i\}),x])=0\;\forall x\in\N\;\forall f:\{1,\ldots,n\}\ra\CC\}$
\end{enumerate}
are equivalent.
\end{lemma}

\begin{proof}
Using the property $\tr(\rho[P,x])=\tr([P,\rho]x)$, we have
\begin{equation}
        \tr(\rho[P,x])=0\;\forall x\in\N\;\;
        \iff\;\;
        \tr([P,\rho]x)=0\;\forall x\in\N\;\;
        \iff\;\;
        [P,\rho]=0.
\end{equation}
Since there is finitely many projections, every function in a $W^*$-algebra $\N$ generated by these projections is their finite sum, and is equivalent to a function on $n$ points.
\end{proof}

\begin{remark}
As shown by Lemma \ref{commutation.as.expectation.lemma}, the weak L\"{u}ders rule \eqref{weak.vonNeumann.Lueders} has no equivalent in the commutative case, because then the condition $\phi([x,y])=0$ is satisfied trivially for arbitrary $x,y\in\N$. Hence, one \textit{cannot} interpret the weak L\"{u}ders rule as a noncommutative \textit{generalisation} of the Bayes--Laplace rule.
\end{remark}

\begin{remark}
Now we will generalise Theorem \ref{vonNeumann.Lueders.from.entropy.theorem} to arbitrary $W^*$-algebras $\N$ and $\rho,\psi,\phi,\omega\in\M(\N)=\N^+_{\star1}$. In order to prove this theorem, we will need also to generalise Lemma \ref{tracelog.triangle.lemma}, Proposition \ref{small.theorem}, and Lemma \ref{lemma.PiPj}. This will be provided, respectively, by Lemma \ref{gen.tracelog.triangle.lemma}, Proposition \ref{gen.small.theorem}, and Lemma \ref{gen.lemma.PiPj}. Our main tool will be the expression for $D_1(\psi,\phi)=D_0(\phi,\psi)$ in terms of Connes' cocycle, introduced by Petz in \cite{Petz:1986:properties}. 
\end{remark}

\begin{lemma}\label{gen.tracelog.triangle.lemma}
Given $\rho,\psi\in\M(\N)$, consider a GNS representation $(\H_\psi,\pi_\psi,\Omega_\psi)$. If $\Omega_\psi\in\dom(\log(\Delta_{\rho,\psi}))$, then the triangle equality for $\psi$ at $\rho$ with respect to $D_0$ is equivalent to
\begin{equation}
        \forall\psi\in\M(\N)\;\;
        \exists\rho\in\Q\;\;
        \forall\phi\in\Q\;\;\;
        \ii\lim_{t\ra^+0}\left(
                \psi(\II-\Connes{\phi}{\rho}{t})
        \right)
        =\ii\lim_{t\ra^+0}\left(
                \rho(\II-\Connes{\phi}{\rho}{t})
        \right).
\end{equation}
\end{lemma}

\begin{proof}
If $\Omega_\psi\in\dom(\log(\Delta_{\rho,\psi}))$, then
\begin{equation}
        \lim_{t\ra^+0}
        \n{\frac{1}{t}(\Delta_{\rho,\psi}^{\ii t}-\II)\Omega_\psi}_{\H_\psi}
        =\lambda\leq\infty.
\end{equation}
Moreover, given
\begin{equation}
        x(t):=\II-\pi_\psi(\Connes{\phi}{\rho}{t}),
\end{equation}
we have
\begin{align}
        \forall t\in\RR\;\;\n{x(t)}_{\H_\psi}&\leq 2,\\
        \lim_{t\ra^+0}x(t)\Omega_\psi&=0,\\
        \lim_{t\ra^+0}\n{x(t)\Omega_\psi}_{\H_\psi}&=0.
\end{align}
Hence,
\begin{equation}
        \lim_{t\ra^+0}
        \ab{\s{x(t)\Omega_\psi,%
        \frac{1}{t}(\Delta_{\rho,\psi}^{\ii t}-\II)\Omega_\psi}_\psi}
        \leq
        \lim_{t\ra^+0}
        \n{x(t)\Omega_\psi}_{\H_\psi}\n{\frac{1}{t}(\Delta_{\rho,\psi}^{\ii t}-\II)\Omega_\psi}_{\H_\psi}
        \leq 0\cdot\lambda 
        =0.
\end{equation}
So,
\begin{align}
        0
        &
        =\lim_{t\ra^+0}\s{\Omega_\psi,\frac{1}{t}%
        \left(
                \II-\pi_\psi(\Connes{\phi}{\rho}{t})
        \right)
        \left(
                \Delta^{\ii t}_{\rho,\psi}-\II
        \right)
        \Omega_\psi}_\psi
        \nonumber
        \\
        &
        =\lim_{t\ra^+0}\s{\Omega_\psi,\frac{1}{t}%
        \left(
                \II-\pi_\psi(\Connes{\phi}{\rho}{t})
        \right)
        \left(
                \pi_\psi(\Connes{\rho}{\psi}{t})-\II
        \right)
        \Omega_\psi}_\psi
        \nonumber
        \\
        &
        =\lim_{t\ra^+0}\frac{1}{t}\psi
        \left(
                (\II-\Connes{\phi}{\rho}{t})
                (\Connes{\rho}{\psi}{t}-\II)
        \right),
\end{align}
where the second equation follows from the property $\Delta^{\ii t}_{\rho,\psi}\xi_\pi(\psi)=\Connes{\rho}{\psi}{t}\xi_\pi(\psi)$ $\forall t\in\RR$ for a standard representative $\xi_\pi(\psi)\in\stdcone$ of $\psi$ in a standard representation $(\H,\pi,J,\stdcone)$, which in this case is given by the cyclic vector $\Omega_\psi$ of the GNS representation $(\H_\psi,\pi_\psi,\Omega_\psi)$. The triangle equality \eqref{triangle.equality} reads
\begin{equation}
        \ii\lim_{t\ra^+0}\frac{\psi}{t}\left(\Connes{\rho}{\psi}{t}-\II\right)+
        \ii\lim_{t\ra^+0}\frac{\rho}{t}\left(\Connes{\phi}{\rho}{t}-\II\right)=
        \ii\lim_{t\ra^+0}\frac{\psi}{t}\left(\Connes{\phi}{\psi}{t}-\II\right),
\end{equation}
and is equivalent to
\begin{equation}
        \ii\lim_{t\ra^+0}\frac{\psi}{t}
                \left(
                        \Connes{\rho}{\psi}{t}-
                        \Connes{\phi}{\psi}{t}
                \right)
        =
        \ii\lim_{t\ra^+0}\frac{\rho}{t}
                \left(
                        \II-\Connes{\phi}{\rho}{t}
                \right)
\end{equation}
It remains to calculate
\begin{align}
        &
        \ii\lim_{t\ra^+0}\frac{1}{t}
        \left(
                \psi(\Connes{\rho}{\psi}{t}-\Connes{\phi}{\psi}{t})
        \right)=
        \nonumber
        \\
        &
        \ii\lim_{t\ra^+0}\frac{1}{t}
        \left(
                \psi(\Connes{\rho}{\psi}{t}-\Connes{\phi}{\rho}{t}\Connes{\rho}{\psi}{t})
        \right)=
        \nonumber
        \\
        &
        \ii\lim_{t\ra^+0}\frac{1}{t}
        \left(
                \psi((\II-\Connes{\phi}{\rho}{t})\Connes{\rho}{\psi}{t})
        \right)=
        \nonumber
        \\
        &
        \ii\lim_{t\ra^+0}\frac{1}{t}
        \left(
                \psi(\II-\Connes{\phi}{\rho}{t})+
                \psi((\II-\Connes{\phi}{\rho}{t})(\Connes{\rho}{\psi}{t}-\II))
        \right)=
        \nonumber
        \\
        &
        \ii\lim_{t\ra^+0}\frac{1}{t}
        \left(
                \psi(\II-\Connes{\phi}{\rho}{t})
        \right).
\end{align}
\end{proof}

\begin{lemma}\label{block.diagonality.of.phi}
If $P\in\Proj(\N)$ and $\phi\in\M(\N)$ then $\phi([P,x])=0$ $\forall x\in\N$ if{}f $\phi$ is block-diagonal, that is, if{}f $\phi(Px(\II-P))=\phi((\II-P)xP)=0$.
\end{lemma}

\begin{proof}
If $\phi([P,x])=0$, then $\phi(Px)=\phi(xP)$, so, for $x=:y(\II-P)$,
\begin{equation}
        \phi(Py(\II-P))=\phi(y(\II-P)P)=0,
\end{equation}
and similarly $\phi((\II-P)yP)=0$. Conversely, every $y\in\N$ has the form 
\begin{equation}
        y=PyP+(\II-P)yP+Py(\II-P)+(\II-P)y(\II-P).
\end{equation}
Hence,
\begin{align}
        [P,y]&=PyP+Py(\II-P)-PyP-(\II-P)yP=Py(\II-P)yP,\\
        \phi([P,y])&=\phi(Py(\II-P)-(\II-P)yP)=0.
\end{align}
\end{proof}

\begin{proposition}\label{gen.small.theorem}
Let $\Q=\{\omega\in\M(\N)\mid\omega([P,x])=0\;\forall x\in\N\}$, and let $\rho=\psi(P\,\cdot\,P)+\psi((\II-P)\,\cdot\,(\II-P))$, where $P\in\Proj(\N)$. Then
\begin{equation}
        \forall\phi\in\Q\;\;D_0(\phi,\psi)=D_0(\phi,\rho)+D_0(\rho,\psi).
\end{equation}
\end{proposition}

\begin{proof}
We will use block decomposition of $\N$ into 
\begin{equation}
        \N=\left(\begin{array}{cc}
                P\N P
                &
                P\N(\II-P)
                \\
                (\II-P)\N P
                &
                (\II-P)\N(\II-P)
        \end{array}\right),
\end{equation}
together with the corresponding notation
\begin{align}
        \forall x\in\N\;\;x&=
        PxP+Px(\II-P)+(\II-P)xP+(\II-P)x(\II-P)
        \nonumber\\&=:
        x_{11}+x_{21}+x_{12}+x_{22}=:
        \left(\begin{array}{cc}
                x_{11}
                &
                x_{21}
                \\
                x_{12}
                &
                x_{22}
        \end{array}\right),
        \\
        \forall x\in\N\;\forall\phi\in\N_\star^+\;\;\phi(x)&=
        \phi\left(\begin{array}{cc}
                x_{11}
                &
                x_{21}
                \\
                x_{12}
                &
                x_{22}
        \end{array}\right)=:
        \phi_{11}(x_{11})+\phi_{21}(x_{21})+\phi_{12}(x_{12})+\phi_{22}(x_{22})
        \nonumber\\
        &=:\left(\begin{array}{cc}
                \phi_{11}
                &
                \phi_{21}
                \\
                \phi_{12}
                &
                \phi_{22}
        \end{array}\right)(x).
\end{align}
By Lemma \ref{block.diagonality.of.phi}, the states $\phi,\rho$ are block-diagonal, 
\begin{equation}
\phi=\left(\begin{array}{cc}\phi_{11}&0\\0&\phi_{22}\end{array}\right),\;\;\;
\rho=\left(\begin{array}{cc}\rho_{11}&0\\0&\rho_{22}\end{array}\right).
\end{equation} 

Recall from \eqref{relative.modular.for.normal.states}-\eqref{RR.Delta.relative} that for a given standard representation of $\N$ on a Hilbert space $\H$, and $\xi\in\H$ such that $\xi\oc[\N\Omega_\rho]$, the relative modular operator $\Delta_{\phi,\rho}$ is defined as $\Delta_{\phi,\rho}:=R^*_{\phi,\rho}\bar{R}_{\phi,\rho}$, where \cite{Araki:1977:relative:entropy:II}
\begin{equation}
        R_{\phi,\rho}(x\Omega_\rho+\xi)=(\supp(\phi))x^*\Omega_\phi.
\end{equation}
From 
\begin{equation}
        \s{x\Omega_\rho,P\xi}_\H=\s{Px\Omega_\rho,\xi}_\H=0\;\;\forall x\in\N
\end{equation}
we have $P\xi\oc[\N\Omega_\rho]$, and, analogously, $(\II-P)\xi\oc[\N\Omega_\rho]$, so, because $P\in\N$, it preserves the dense domain. From
\begin{align}
\s{(\II-P)(x_1\Omega_\rho+\xi_1),\Delta_{\phi,\rho}P(x_2\Omega_\rho+\xi_2)}_\H&=\s{\bar{R}_{\phi,\rho}(\II-P)(x_1\Omega_\rho+\xi_1),\bar{R}_{\phi,\rho}P(x_2\Omega_\rho+\xi_2)}_\H
\nonumber\\&=
\s{\bar{R}_{\phi,\rho}((\II-P)x_1\Omega_\rho+(\II-P)\xi_1),\bar{R}_{\phi,\rho}(Px_2\Omega_\rho+P\xi_2)}_\H
\nonumber\\&=
\s{\supp(\phi)x_1^*(\II-P)\Omega_\phi,\supp(\phi)x_2^*P\Omega_\phi}_\H
\nonumber\\&=
\phi((\II-P)x_1\supp(\phi)x_2^*P)
\nonumber\\&=0
\end{align}
we have that $[P,\Delta_{\phi,\rho}]=0$ on the dense domain. Since $P$ is bounded, it follows that $\Delta_{\phi,\rho}$ preserves the decomposition $\ran(P)\oplus\ran(\II-P)$. Hence, $\Delta_{\phi,\rho}$ is block diagonal. The same holds for $\Delta_{\rho}=\Delta_{\rho,\rho}$. Thus, $\Connes{\phi}{\rho}{t}=\Delta^{\ii t}_{\phi,\rho}\Delta^{-\ii t}_\rho$ is also block diagonal,
\begin{align}
        \Connes{\phi}{\rho}{t}&=\left(\begin{array}{cc}
        (\Connes{\phi}{\rho}{t})_{11}
        &
        0
        \\
        0
        &
        (\Connes{\phi}{\rho}{t})_{22}
        \end{array}\right),\\
        \II-\Connes{\phi}{\rho}{t}&=
        P(\II-\Connes{\phi}{\rho}{t})P+
        (\II-P)(\II-\Connes{\phi}{\rho}{t})(\II-P).
\end{align}

From $\psi(PxP+(\II-P)x(\II-P))=\rho(PxP+(\II-P)x(\II-P))\;\forall x\in\N$, and
we obtain
\begin{equation}
        \psi(\II-\Connes{\phi}{\rho}{t})=\rho(\II-\Connes{\phi}{\rho}{t}).
\end{equation}
Hence, by Lemma \ref{gen.tracelog.triangle.lemma}, triangle equality holds if $\Omega_\psi\in\dom(\log(\Delta_{\rho,\psi}))$. It remains to check whether this condition is satisfied. Araki and Masuda \cite{Araki:Masuda:1982} and Donald \cite{Donald:1990} prove this relation, but under stronger conditions (Araki and Masuda assume that $\rho$ is faithful, while Donald assumes also faithfulness of $\psi$), so we need to provide more general proof. Let $\xi\oc[\N\Omega_\psi]$, and consider operators $R_\psi$, $R_{\rho,\psi}$ and $A$ given by
\begin{align}
        R_\psi(x\Omega_\psi+\xi)&:=(\supp(\psi))x^*\Omega_\psi,\\
        R_{\rho,\psi}(x\Omega_\psi+\xi)&:=(\supp(\psi))x^*\Omega_\psi,\\
        A(x\Omega_\psi)&:=x\Omega_\rho.
\end{align}
Then
\begin{equation}
        \n{A(x\Omega_\psi)}^2_{\H_\psi}=\rho(x^*x)\leq\lambda\psi(x^*x)=\lambda\n{x\Omega_\psi}^2_{\H_\psi},
\end{equation}
where $\lambda\in\RR^+$, so $A$ is bounded. From $R_{\rho,\psi}=AR_\psi$, $J_\psi\Delta^{-1/2}_\psi=R_\psi$, and $J_{\rho,\psi}\Delta^{-1/2}_{\rho,\psi}=R_{\rho,\psi}$ we obtain
\begin{equation}
        \Delta^{-1/2}_{\rho,\psi}=J^{-1}_{\rho,\psi}AJ_\psi\Delta^{-1/2}_\psi,
\end{equation}
where $J^{-1}_{\rho,\psi}AJ_\psi$ is bounded. So
\begin{equation}
\Omega_\psi\in\dom(\Delta^{-1/2}_\psi)\subseteq\dom(\Delta^{-1/2}_{\rho,\psi}).
\end{equation}
Araki \cite{Araki:1976:relative:entropy:I} proved that $\Omega_\rho=(\Delta_{\rho,\psi})^{1/2}\Delta_{\psi}$ and $\Delta_\psi\in\dom(\Delta^{1/2}_{\rho,\psi})$. Hence
\begin{equation}
        \Omega_\psi\in
        \dom(\Delta^{1/2}_{\rho,\psi})\cap\dom(\Delta^{-1/2}_{\rho,\psi}).
\end{equation}
On the other hand,
\begin{equation}
\exists\lambda\in\RR^+\;\;\forall\gamma\in\RR^+\;\;\ab{\log(\gamma)}\leq\lambda\max(\gamma^{1/2},\gamma^{-1/2}),
\end{equation}
so, for every positive operator $x$ on a Hilbert space $\H$,
\begin{equation}
        \dom(\log x)\supseteq\dom(x^{1/2})\cap\dom(x^{-1/2}).
\end{equation}
It follows that $\Omega_\psi\in\dom(\log(\Delta_{\rho,\psi}))$.
\end{proof}

\begin{lemma}\label{gen.lemma.PiPj}
Given $P_i,Pj\in\Proj(\N)$, let $\Q_k=\{\omega\in\M(\N)\mid\omega([P_k,x])=0\;\forall x\in\N\}$ for $k\in\{i,j\}$, and $[P_i,P_j]=0$. If $\psi\in\Q_j$, then
\begin{equation}
        \arginff{\phi\in\Q_i}{D_0(\phi,\psi)}\in\Q_j.
\end{equation}
\end{lemma}

\begin{proof}
From Proposition \ref{gen.small.theorem} it follows that
\begin{equation}
        \arginff{\phi\in\Q_i}{D_0(\phi,\psi)}=\psi(P(\,\cdot\,)P)+\psi((\II-P)(\,\cdot\,)(\II-P)).
\label{arginf.D0.vNLu.lemma}
\end{equation}
Denote  the right hand side of \eqref{arginf.D0.vNLu.lemma} by $\varphi$. We need to check that $\psi([P_j,x])=0$ $\forall x\in\N$ $\limp$ $\varphi([P_j,x])=0$ $\forall x\in\N$. But this follows from
\begin{align}
        \varphi([P_j,x])&=\psi(P_i[P_j,x]P_i+(\II-P_i)[P_j,x](\II-P_i))\nonumber\\
        &=\psi([P_j,P_ixP_i+(\II-P_i)x(\II-P_i)])\nonumber\\
        &=0.
\end{align}
\end{proof}

\begin{theorem}\label{gen.vonNeumann.Lueders.from.entropy.theorem}
If $\{P_i\}_{i=1}^n\subseteq\Proj(\N)$ satisfies $P_iP_j=\dirac_{ij}P_i$ $\forall i,j\in\{1,\ldots,n\}$, $\sum_{i=1}^nP_i=\II$, and 
\begin{equation}
        \Q_{\mathrm{L}}=\{\omega\in\M(\N)\mid \omega([P_i,x])=0\;\forall x\in\N\;\forall i\in\{1,\ldots,n\}\},
\end{equation}
then 
\begin{equation}
        \PPP^{D_0}_{\Q_{\mathrm{L}}}(\psi)=\sum_{i=1}^nP_i\psi P_i.
\label{gen.weak.vonNeumann.Lueders}
\end{equation}
\end{theorem}

\begin{proof}
The same as for Theorem \ref{vonNeumann.Lueders.from.entropy.theorem}, with the first step of induction satisfied by Proposition \ref{gen.small.theorem}, and using Lemma \ref{gen.lemma.PiPj} instead of Lemma \ref{lemma.PiPj}.
\end{proof}

\begin{theorem}\label{weighted.vNLu.theorem}
Let $\N=\BH$ and $\M(\N)=\BH_{\star1}^+$. If $\psi\in\M(\N)$, $\{P_i\}_{i=1}^n\subseteq\Proj(\N)$ satisfy $P_iP_j=\dirac_{ij}P_i$ $\forall i,j\in\{1,\ldots,n\}$, $\sum_{i=1}^nP_i=\II$, and $\tr(\psi P_i)\neq0$ $\forall i\in\{1,\ldots,n\}$, if $\lambda_i\in\RR$ satisfy $\sum_{i=1}^n\lambda_i=1$, and
\begin{equation}
        \Q_{\mathrm{qJ}}:=\{\omega\in\schatten_1(\H)^+_1\mid\tr(\omega P_i)=\lambda_i\}
\label{Q.qJ.constraint}
\end{equation}
then
\begin{equation}
        \PPP^{D_0}_{\Q_{\mathrm{qJ}}}(\psi)=
        \sum_{i=1}^n\lambda_i\frac{P_i\psi P_i}{\tr(\psi P_i)}.
\label{weighted.vNLu.eqn}
\end{equation}
\end{theorem}

\begin{proof}
Let $\rho=\bigoplus_{i=1}^n\rho_i$, $\phi=\bigoplus_{i=1}^n\phi_i$, and $\rho_i=\lambda_i\frac{P_i\psi P_i}{\tr(P_i\psi_i P_i)}$, $\tr(\phi_i)=\lambda_i$ $\forall i\in\{1,\ldots,n\}$. Then, using the fact that each function of block diagonal matrices is block diagonal, we obtain
\begin{equation}
        \log\rho-\log\phi=\bigoplus_{i=1}^n\left(\log\rho_i-\log\phi_i\right)=:\bigoplus_{i=1}^nC_i.
\end{equation}
Using $\tr(\frac{\rho_i}{\lambda_i})=1=\tr(\frac{\phi_i}{\lambda_i})$, we obtain
\begin{align}
        \tr\left(\psi\bigoplus_{i=1}^nC_i\right)&=\sum_{i=1}^n\tr\left(P_i\psi P_iC_i\right)=\sum_{i=1}^n\tr\left(P_i\psi P_i\left(\log\frac{\rho_i}{\lambda_i}-\log\frac{\phi_i}{\lambda_i}\right)\right)\nonumber\\
        &=\sum_{i=1}^n\tr(P_i\psi P_i)\tr\left(\frac{P_i\psi P_i}{\tr(P_i\psi P_i)}\left(\log\frac{\rho_i}{\lambda_i}-\log\frac{\phi_i}{\lambda_i}\right)\right)\nonumber\\
        &=\sum_{i=1}^n\tr(P_i\psi P_i)\tr\left(\widetilde{\rho}_i\left(\log\widetilde{\rho}_i-\log\widetilde{\phi}_i\right)\right)\nonumber\\
        &=\sum_{i=1}^n\tr(P_i\psi P_i)D_0(\widetilde{\phi}_i,\widetilde{\rho}_i)\geq0,\label{tilde.D.zero.oplus}
\end{align}
where $\widetilde{\rho}_i:=\frac{\rho_i}{\lambda_i}$, $\widetilde{\phi}_i:=\frac{\phi_i}{\lambda_i}$, and equality is attained if{}f $\widetilde{\phi}_i=\widetilde{\rho}_i$. Hence,
\begin{align}
        \tr(\psi(\log\rho-\log\phi))&\geq0,\\
        \tr(\psi(\log\psi-\log\rho))&\leq\tr(\psi(\log\psi-\log\phi)).
\end{align}
From the condition for equality in \eqref{tilde.D.zero.oplus}, it follows that $\rho$ is the unique minimiser of $D_0(\cdot,\psi)$.
\end{proof}

\begin{remark}\label{remark.strong.vNLu}
The strong L\"{u}ders rule \eqref{strong.Lueders.rule} can be obtained from minimisation of $D_0$ by two different methods. First method amounts to applying quantum Jeffrey's rule \eqref{weighted.vNLu.eqn} and taking the limit $\lambda_2,\ldots,\lambda_n\ra0$,
\begin{equation}
        \lim_{\lambda_2,\ldots,\lambda_n\ra0}\PPP^{D_0}_{\Q_{\mathrm{qJ}}}(\psi)=\frac{P_1\psi P_1}{\tr(\psi P_1)}.
\label{strong.vN.rule.from.limit}
\end{equation}
Note that \eqref{weighted.vNLu.eqn} is a weakly continuous function of $\lambda_i$. The limit \eqref{strong.vN.rule.from.limit} is also weakly continuous. Hence, the strong L\"{u}ders rule can be considered as a weakly continuous limit of an entropic projection $\PPP^{D_0}$.
\end{remark}

\begin{remark}
Despite the observation carried in Remark \ref{strongLuedersfromqJr.remark}, the direct derivation of \eqref{strong.Lueders.rule} along the lines of Theorem \ref{weighted.vNLu.theorem} with the initial constraints $\lambda_2=\ldots=\lambda_n=0$ is impossible, because in such case the necessary assumptions for this theorem do not hold. More precisely, the states $\omega$ that satisfy
\begin{equation}
        -\tr\left(\omega\log(P\rho P-(\II-P)\rho(\II-P))\right)<\infty
\end{equation}
do not exist if $\tr(\rho P)=0$. This follows from
\begin{align}
        -\tr\left(\omega\left(
                \log|_{P\N P}(P\rho P)+\log|_{(\II-P)\rho(\II-P)}(\II-P)\rho(\II-P)
        \right)\right)&=
        -\tr\left(
                \omega\log|_{P\N P}(P\rho P)
        \right)\nonumber\\&=
        -\tr(\omega\log0)=
        -\infty.
\end{align}
This situation can be improved by `regularisation' of the difference of two distance functionals, using Connes' cocycle with respect to some well-behaved `reference' functional $\omega_0\in\M(\N)$. 
The natural choice in the case of strong L\"{u}ders' rule is $\omega_0=\rho|_{P\N P}$. (Donald \cite{Donald:1986,Donald:1987:further:results} also introduces a distance functional that is dependent on the choice of a subset of $C^*$-algebra, but his motivation as well as the resulting mathematical construction differ from ours.) Hence, in order to show that the strong L\"{u}ders rule is an entropic projection without any limits involved, we will use an analogue of a Theorem \ref{gen.vonNeumann.Lueders.from.entropy.theorem} for a single projection $P$ and a \textit{regularised} distance functional $D_0^P(\phi,\psi)$, defined as a restriction of $D_0(\phi,\psi)$ to a subspace $P\N P$.
\end{remark}

\begin{lemma}
\begin{equation}
        D^P_0(\omega,\rho):=
        \ii\lim_{t\ra^+0}\rho|_{P\N P}\left(
                \Connes{\omega|_{P\N P}}{\rho|_{P\N P}}{t}
        -\II|_{P\N P}\right).
\end{equation}
is a distance functional on $(P\N P)^+_{\star1}$, that is,
\begin{align} 
D^P_0(\omega,\rho)&\geq0\;\;\forall\rho,\omega\in(P\N P)^+_{\star1},\\
D^P_0(\omega,\rho)&=0\iff\rho|_{P\N P}=\omega|_{P\N P}.
\end{align}
\end{lemma}

\begin{proof}
Follows directly from the definition and properties of $D_0|_{\N^+_{\star1}}$.
\end{proof}


\begin{lemma}\label{p.omega.lemma}
If $\tr(\omega P)=\tr(\omega)=1$ then $[\omega,P]=0$ and $\omega=P\omega P$.
\end{lemma}

\begin{proof}
We need to check that $\tr((\II-P)\omega P)=0$, which is equivalent to $\omega((\II-P)xP)=0$ $\forall x\in\N$. But this follows from
\begin{equation}
        |\omega((\II-P)x)|^2\leq\omega(x^*x)\omega((\II-P)^2)=0.
\end{equation}
Similarly $\omega(x(\II-P))=0$ $\forall x\in\N$.
\end{proof}

\begin{theorem}\label{regularised.MRE.vN.theorem}
If $\Q_{\mathrm{sL}}:=\{\omega\in\M(\N)\mid\tr(\omega P)=\tr(\omega)=1\}$, and $\psi\in\M(\N)$ satisfies $\tr(P\psi)\neq0$, then
\begin{equation}
        \PPP^{D^P_0}_{\Q_{\mathrm{sL}}}(\psi)\equiv\arginff{\phi\in\Q_{\mathrm{sL}}}{D_0^P(\phi,\psi)}=\frac{P\psi P}{\tr(P\psi)}.
\end{equation}
\end{theorem}

\begin{proof}
Follows from Theorem \ref{weighted.vNLu.theorem}, applied to $(P\N P)^+_{\star1}$, if we notice that 
\begin{equation}
        \Q_{\mathrm{sL}}=\{\omega\in\M(\N)\mid\tr(\omega)=1,\;\omega=P\omega P\},
\label{QsL.constraint}
\end{equation}
which follows from Lemma \ref{p.omega.lemma}.
\end{proof}
\section{Minimisation of quantum metrical distances\label{transition.correlation.section}}
In the Hilbert space based quantum mechanics a \df{transition probability} of $\xi_0,\xi_1\in\H$ such that $\n{\xi_0}=\n{\xi_1}=1$ is defined as\rpktarget{TRANSPROB}
\begin{equation}
        \TP_\H(\xi_1,\xi_0):=\ab{\s{\xi_1,\xi_0}}^2\in[0,1].
\label{trans.prob}
\end{equation}
Different generalisations of the notion of transition probability to the case of quantum states over $W^*$-algebras are possible. The two most important are: the \df{Raggio transition probability} \cite{Raggio:1982}\rpktarget{RAGGIO}
\begin{equation}
        \TP_{\mathrm{R}}(\phi,\psi):=\s{\xi_\pi(\phi),\xi_\pi(\psi)}_\H=\frac{1}{2}\left(\phi(\II)+\psi(\II)-\n{\xi_\pi(\phi)-\xi_\pi(\psi)}^2_\H\right)\;\;\;\forall\phi,\psi\in\N_\star^+,
\label{Raggio.trans.prob}
\end{equation}
where $(\H,\pi,J,\stdcone)$ is a standard representation of $\N$, and the \df{Cantoni--Uhlmann transition probability} \cite{Cantoni:1975,Uhlmann:1976}\rpktarget{CANULH}
\begin{equation}
        \TP_{\mathrm{CU}}(\phi,\psi):=\sup_{(\H,\pi)}\left\{\ab{\s{\zeta_\pi(\phi),\zeta_\pi(\psi)}_\H}^2\right\}\;\;\;\forall\phi,\psi\in\N_\star^+,
\label{Cantoni.Uhlmann.trans.prob}
\end{equation}
where $\zeta_\pi(\omega)\in\H$ is defined by $\omega(x)=\s{\zeta_\pi(\omega),\pi(x)\zeta_\pi(\omega)}_\H$ $\forall x\in\N$ for some representation $(\H,\pi)$ of $\N$, and the supremum varies over all possible representations. For the comparison of \eqref{Raggio.trans.prob} with \eqref{Cantoni.Uhlmann.trans.prob} and with some other possibilities, see \cite{Raggio:1982,Raggio:1984,Alberti:Uhlmann:1984,Yamagami:2008,Yamagami:2010}. 

From the geometric perspective it is worth noting that \eqref{Raggio.trans.prob} is bijectively related to the distance on $\N_\star^+$ defined by the norm of $\H\iso L_2(\N)$,
\begin{equation}
        d_{L_2(\N)}(\phi,\psi)=\n{\xi_{\pi_\N}(\phi)-\xi_{\pi_\N}(\psi)}_{L_2(\N)}=\sqrt{\phi(\II)+\psi(\II)-2\TP_{\mathrm{R}}(\phi,\psi)},
\end{equation}
as well as to the $D_{1/2}$ distance \eqref{D.one.half.distance} on $\N_\star^+$,
\begin{align}
        D_{1/2}(\phi,\psi)&=2(\phi+\psi)(\II)-4\s{\xi_\pi(\phi),\xi_\pi(\psi)}_\H\nonumber\\
        &=2\left(\phi(\II)+\psi(\II)\right)-4\TP_{\mathrm{R}}(\phi,\psi)=2\n{\xi_{\pi_\N}(\phi)-\xi_{\pi_\N}(\psi)}_{L_2(\N)}^2,
\label{D.onehalf.TP.R}  
\end{align}
while \eqref{Cantoni.Uhlmann.trans.prob} is bijectively related to the \df{Bures distance} on $\N_\star^+$ \cite{Bures:1969} (cf. also \cite{Araki:1972,Araki:1974:modular:conjugation}) defined by\rpktarget{DBURES.ZWEI}
\begin{equation}
        d_{\mathrm{Bures}}(\phi,\psi)=\inf_{(\H,\pi)}\left\{\n{\zeta_\pi(\phi)-\zeta_\pi(\psi)}_\H\right\}=\sqrt{\phi(\II)+\psi(\II)-2\sqrt{\TP_{\mathrm{CU}}(\phi,\psi)}},
\label{Bures.distance}
\end{equation}
where $\zeta_\pi$ is defined as above, and $\inf$ varies over the same range as $\sup$ in \eqref{Cantoni.Uhlmann.trans.prob}. In the notation above we have used freely the unitary equivalence between any standard representation $(\H,\pi,J,\stdcone)$ of a $W^*$-algebra $\N$ and its canonical representation $(L_2(\N),\pi_\N,J_\N,L_2(\N)^+)$, see \cite{Kosaki:1980:PhD,Kostecki:2013}. Both $d_{L_2(\N)}(\phi,\psi)$ and $d_{\mathrm{Bures}}(\phi,\psi)$ are metrical distances.

Herbut \cite{Herbut:1969} proved that 
\begin{equation}
	\PPP^{d_{L_2(\N)}}_{\Q_{\mathrm{L}}}(\psi)=\sum_{i=1}^nP_i\rho_\psi P_i
\label{Herbut.equation}
\end{equation}
for $\N=\BH$, $\rho_\phi,\rho_\psi\in\schatten_1(\H)^+$, and $\Q_{\mathrm{L}}$ given by \eqref{QL.constraint}. This derivation of the weak L\"{u}ders rule was the first result of this type in the literature. Minimisation of the same distance function, but under constraints of the type $\tr(\rho_\phi x)=\lambda$ with $x\in\BH^\sa$ and $\lambda\in\RR$, was later considered in \cite{Dieks:Veltkamp:1983}, however with no general results. 

Let $\N_0\subseteq\N$ be the $W^*$-subalgebras of $\BH$, let $T\in\N^+$ be invertible with $0<T\leq\II$, and let $\psi\in\N^+_{\star1}$. Marchand and collaborators \cite{Marchand:1977,Marchand:Wyss:1977,Gudder:Marchand:Wyss:1979} considered a quantum inference problem based on
\begin{equation}
	\arginff{\omega\in K}{d_{\mathrm{Bures}}(\omega,\psi)}=\argsupp{\omega\in K}{\TP_{\mathrm{CU}}(\omega,\psi)},
\end{equation}
with $K=\{\phi\in\N^+_{\star1}\mid\phi|_{\N_0}=\psi(T\cdot T)\}$. The algebra $\N_0$ is interpreted as representing operators subjected to a ``partial measurement'', while $T$ is a noncommutative analogue of the Radon--Nikod\'{y}m quotient that follows from Sakai's theorem \cite{Sakai:1965} and can be thought of as a generalisation of a projection.\footnote{For a generalisation to a setting based on $C^*$-algebras, see \cite{Gudder:1980}. For a generalisation that does not require a subset $\N_0$ to be a $W^*$-algebra, see \cite{Marchand:1983}.} They derived in \cite{Benoist:Marchand:Yourgrau:1977,Benoist:Marchand:1979,Marchand:1983} several different ``post-measurement'' states, dependent on the choice of $\N_0$, $T$, and initial correlations in $\psi$. 

It was shown by Raggio \cite{Raggio:1984} that the strong L\"{u}ders rule \eqref{strong.Lueders.rule} can be directly derived as a special case of constrained \textit{maximisation} of the Cantoni--Uhlmann transition probability \eqref{Cantoni.Uhlmann.trans.prob}. Let $Y$ be a convex subset of a real topological vector space $X$. A subset $F\subseteq Y$ is called a \df{face} if{}f 
\begin{equation}
        \forall x\in F\;\;
        \exists n\in\NN\;\;\left(
                \exists\{\lambda_i\}_{i=1}^n\subseteq\RR^+\;\;
                x=\sum_{i=1}^n\lambda_ix_i,\;\;
                \sum_{i=1}^n\lambda_i=1
        \right)
        \;\;\limp\;\;
                \{x_i\}_{i=1}^n\subseteq F.     
\end{equation}
Let $\N$ be a $W^*$-algebra, and let $K$ be a closed, convex subset of $\N^+_{\star1}$ such that 
\begin{equation}
        \left(\lambda\omega+(1-\lambda)\phi\in K\;\;\forall\lambda\in[0,1]\;\;\limp\;\;\omega,\phi\in K\right)\;\;\forall\omega,\phi\in\N^+_{\star1}.
\end{equation}
Such set is a face in $\N^+_{\star1}$. For each face $K\subseteq\N^+_{\star1}$ there exists a unique $P\in\Proj(\N)$ such that $\omega\in K\iff\omega(P)=1$. If $\Q_P$ is a face in $\N^+_{\star1}$ with a corresponding $P\in\Proj(\N)$, $\psi\in\N^+_{\star1}$, $\N=\BH$, and $\tr(\rho_\psi\,\cdot\,)\equiv\psi$, then \cite{Raggio:1984}
\begin{equation}
        \frac{P\rho_\psi P}{\tr(P\rho_\psi)}=\argsupp{\omega\in\Q_P}{\TP_{\mathrm{CU}}(\omega,\psi)}.
\label{argsup.TP.CU.eq}
\end{equation}
This corresponds to Domotor's observation \cite{Domotor:1985} that the faces in $L_1(\X,\mho(\X),\tmu)^+_1$ form the correct constraints for the Bayes--Laplace rule. On the other hand, \eqref{D.onehalf.TP.R} gives us
\begin{equation}
        \argsupp{\omega\in\Q_P}{\TP_{\mathrm{R}}(\omega,\psi)}=\arginff{\omega\in\Q_P}{D_{1/2}(\omega,\psi)}.
\label{argsup.TP.R.eq}
\end{equation}
From the fact \cite{Raggio:1984} that 
\begin{equation}
        \argsupp{\omega\in\Q_P}{\TP_{\mathrm{R}}(\omega,\psi)}=
        \argsupp{\omega\in\Q_P}{\TP_{\mathrm{CU}}(\omega,\psi)}
\end{equation}
whenever $\N$ is commutative, or $\psi$ is pure, or $\Q_P=\{*\}$, or $P\in\{x\in\N\mid\sigma^\psi_t(x)=x\;\forall t\in\RR\}$ for $\sigma^\psi_t:=\pi_\omega^{-1}(\Delta_{\omega,\omega}^{\ii t}\pi_\omega(x)\Delta_{\omega,\omega}^{-\ii t})$, we can conclude that in any of these cases
\begin{equation}
        \PPP^{D_{1/2}}_{\Q_P}(\psi)=\frac{P\rho_\psi P}{\tr(P\rho_\psi)}.
\label{strong.vN.D.onehalf}
\end{equation}
Thus, under the above conditions, the strong L\"{u}ders rule can be derived as a result of constrained minimisation of $D_{1/2}(\omega,\psi)$. However, these conditions are so restrictive (eliminating e.g. nonpure density operators) that it is more proper to say that the equation \eqref{strong.vN.D.onehalf} expresses a derivation of the strong von Neumann rule. Raggio \cite{Raggio:1984} showed also that, for a general $\psi$ and $K$ the same as in \eqref{argsup.TP.CU.eq}, \eqref{argsup.TP.R.eq} leads to a different result than \eqref{argsup.TP.CU.eq}. This is also the case for a general $\psi$ and $K=\Q_{\mathrm{sL}}$. Nevertheless, we can prove the following: 
\begin{proposition}
For $\psi\in\schatten_1(\H)^+$ and $\Q_{\mathrm{L}}$ given by \eqref{QL.constraint},
\begin{equation}
	\PPP^{D_{1/2}}_{\Q_{\mathrm{L}}}(\psi)=\sum_{i=1}^nP_i\rho_\psi P_i.
\end{equation}
\end{proposition}
\begin{proof}
Follows directly from \eqref{Herbut.equation} and 
\begin{equation}
(d_{L_2(\N)}(\phi,\psi))^2=2D_{1/2}(\phi,\psi).
\end{equation}
\end{proof}

Another result was obtained by Hadjisavvas \cite{Hadjisavvas:1978,Hadjisavvas:1981}, who showed that the strong von Neumann rule for pure $\psi\in\schatten_1(\H)^+_1$ can be derived as 
\begin{equation}
        \PPP^{d_{L_1(\N)}}_{\Q_{\mathrm{sL}}}=\frac{P\rho_\psi P}{\tr(P\rho_\psi)},
\label{Hajisavvas.min.problem}
\end{equation}
where $d_{L_1(\N)}$ is the metrical \df{Jauch--Misra--Gibson--Kronfli distance} \cite{Jauch:Misra:Gibson:1968,Kronfli:1970,Hadjisavvas:1981} on $\N_\star\iso L_1(\N)$,
\begin{equation}
	d_{L_1(\N)}(\phi,\psi):=\frac{1}{2}\n{\phi-\psi}_{\N_\star}.
\label{D.JMGK}
\end{equation}
It is worth noticing that the original definition \cite{Jauch:Misra:Gibson:1968} of $d_{L_1(\N)}$ was provided over the measures on orthonormal orthomodular lattice $\Proj(\BH)$, where it takes a form
\begin{equation}
	\sup_{x\in\Proj(\BH)}\ab{p_1(x)-p_2(x)},
\end{equation}
while the distance minimised in \cite{May:1973,May:Harper:1976,May:1976,May:1979} can be represented in a form
\begin{equation}
	\sup_{x\in\boole}\ab{p_1(x)-p_2(x)},
\label{abs.subtraction.distance}
\end{equation}
where $\boole$ is a boolean algebra\footnote{More precisely, it is a finite boolean algebra that is defined as a Lindenbaum--Tarski algebra of a predicate calculus language.}. In this sense, Hadjisavvas' derivation of the strong von Neumann rule from minimisation of $d_{L_1(\N)}$ is similar to Burris' derivation of Jeffrey's rule as a (nonunique) minimiser of \eqref{abs.subtraction.distance} (reported in \cite{May:Harper:1976,May:1979}\footnote{In \cite{vanFraassen:1981} this result is incorrectly attributed to \cite{Jamison:1974}, despite the clear statement in \cite{May:Harper:1976,May:1979}. Note that this result was proven only for Jeffrey's rule \eqref{Jeffrey.rule} with $n=2$.}). In \cite{Diu:1982} Diu showed that, when applied to nonpure states $\psi$, $\PPP^{d_{L_1(\N)}}_{\Q_{\mathrm{sL}}}(\psi)$ does not lead to the strong L\"{u}ders rule (in \cite{Diu:1983} he extended this result to a more general family of metrical distances on quantum states). In our opinion, Diu's result cannot be used as a general argument against using constrained minimisation of $d_{L_1(\N)}$, or any other information distance, because strong L\"{u}ders' rule is not a uniquely ``correct'' quantum state change rule. However, on the positive side, Diu's result and the similar result by Raggio on $\PPP^{D_{1/2}}_{\Q_{\mathrm{sL}}}(\psi)$ for nonpure $\psi$ exemplify that the choice of an information distance subjected to minimisation preselects the type of possible constraints and the class of admissible results. In this sense, the choice of an information distance (metrical or not) amounts to the choice of a specific convention of inference, which in turn determines some range of possible forms of information dynamics and their output states.
\section{Discussion\label{discussion.section}}
The information state changes conditioned on certainties (yes/no truth values corresponding to elements of a boolean algebra or orthonormal lattice of projections) are definitely not the only way, and also not the most useful way, of defining information dynamics of information states, both statistical and quantum. The possibility of a derivation of the Bayes--Laplace, Jeffrey's, L\"{u}ders', and quantum Jeffrey's rules from minimisation of different distance functionals subjected to various constraints shows the explanatory (semantic) strength of the approach based on entropic projections.

While \textit{in principle} any quantum state change rule (such as Luders' rule, quantum Jeffrey's rule, or some quantum channel\footnote{That is, a completely positive trace preserving map between quantum states.}) is an \textit{ad hoc} postulate, not derived from any other, more fundamental, property of quantum theoretic formalism, \textit{in practice} the choice of a particular form of this rule reflects certain assumptions about the relationship between the knowledge about outcomes of experimental procedure and the contents of a quantum theoretical model. For example, weak L\"{u}ders' rule assumes the specific type of knowledge about the ``measurement result'', requiring specification of the choice of the projection operators $\{P_i\mid i\in I\}\subseteq\BH$. This restricts the allowed form of `experimental evidence' to nonempty convex closed linear subspaces of a commutative $L_2$ space.

A virtue of the approach based on quantum entropic projections is that it allows for a vast generalisation beyond the above restriction, while keeping clear underlying conceptual principles as well as strong mathematical useability. Each information distance functional $D$ expresses a choice of a specific convention on the preferred/relevant and unpreferred/irrelevant aspects of information states: the relevant aspects are those that more strongly participate in the values of $D$ (one can think of a conventional character of a least squares distance, which expresses certain arbitrary preferences regarding the information content of the data). To every choice of an information distance there corresponds a preferred type of constraint (preferred geometric form of `experimental evidence') for which this information distance is capable of obtaining a unique minimum. 
The general way to introduce constraints is then to provide a mapping $\Xi\ra\M(\N)$ from the space $\Xi$ of (`epistemic'/`experimental'/`registration') parameters describing the sets of ``possible outcomes'' into nonempty convex closed subsets of $\M(\N)$. 

This way, as opposed to the Bayes--Laplace rule and L\"{u}ders' rules, in our framework the conditioning is provided not upon the abstract `event' that belongs to a boolean algebra or to an orthomodular lattice, respectively, but upon the value taken in the space $\Xi$ of parameters describing the ``possible outcomes''. This is similar to the semi-spectral (povm/cp-map) approach, which uses linear povm-instruments conditioned upon the spaces of ``possible outcomes'' that belong to $\mho(\X)$. This difference (conditioning upon `quantitative information' as opposed to conditioning upon `event') is the key insight. Note that in the semi-spectral approach there is also no bijection between the space of `abstract events' and the space of ``possible outcomes'': a single effect can correspond to various elements $\Y\in\mho(\X)$ \cite{Braunstein:Caves:1988}: the `eigenstate-eigenvalue link' \cite{Fine:1973} breaks down here. However, as opposed to the semi-spectral approach, our approach completely detaches from the reliance on the use of spectral theory in foundations of quantum theory, allowing for more flexible operational specification of the `experimental evidence', and for deriving ``quantum measurement'' rules from a single underlying principle (entropic projection), which is alternative to quantum channels. In particular, the use of quantum distance as an underlying mathematical structure allows for geometric analysis and justification of the choice of a specific convention of quantum inference/dynamics. This is in contrast with the general lack of clear geometric justification for a choice of a specific quantum channel in the semi-spectral approach.

An extended development of the approach based on quantum entropic projections as an alternative to povm-instruments, and as a replacement for L\"{u}ders' rules, will be carried out in \cite{Kostecki:2014:towards}. For an alternative derivation of L\"{u}ders' rules from quantum entropic projections based on $D_0$, see \cite{HKK:2014}.

What does it all mean for quantum bayesianism? Quite often `bayesianism' is understood as a subjective interpretation of probability equipped with the requirement of using the Bayes--Laplace rule for the purpose of changing probabilities due to learning new information. However, the appearance of other updating rules, such as Jeffrey's rule, Field's rule \cite{Field:1978}, and constrained maximisation of the WGKL distance \cite{Kullback:1959,Hobson:Cheng:1973,Johnson:1979,Williams:1980}, has undermined the universality of a second component. The usual perspective on the meaning of `quantum bayesianism' starts from the semi-spectral approach to mathematical foundations of quantum theory, and aims at recasting (some suitable class of) povm-instruments as (a modified form of) the Bayes--Laplace rule, while keeping the subjective interpretation of probabilities. Such perspective assumes that probability theory and spectral theory should be fundamental constitutents of quantum theory. Moreover, it also does not provide justification for using povm-instruments: the mathematical foundations are just taken for granted. The results contained in this paper are intended to serve as a guideline (via quantum Jeffrey's rule) and as a backwards compatibility proof (via recovery of L\"{u}ders' rules) for an alternative approach to the mathematical and conceptual meaning of `quantum bayesianism'. According to our point of view, quantum states should be used as carriers of intersubjective knowledge on their own mathematical right (as elements of a noncommutative $L_1$ space), without reference to probabilities (elements of a commutative $L_1$ space), while the processes of inductive inference (information dynamics) can be fruitfully modelled using quantum entropic projections as a nonlinear alternative to the Bayes--Laplace rule and linear povm-instruments. See \cite{Kostecki:2014:towards} for a detailed account.
\section*{Acknowledgments}
{\small I would like to thank Wojciech Kami\'{n}ski for numerous discussions, comments, and insights that have strongly shaped this work. I thank also Carlos Guedes, Frank Hellmann, and Stanis{\l}aw Woronowicz for discussions at the early stage of this work, Stanley Burris and Nicolas Hadjisavvas for correspondence, Patrick Coles for informing about \cite{MPSVW:2010,Coles:2012}, and Bianca Dittrich for a partial support. This research was supported in part by Perimeter Institute for Theoretical Physics. Research at Perimeter Institute is supported by the Government of Canada through Industry Canada and by the Province of Ontario through the Ministry of Research and Innovation. This research was also partially financed by the National Science Center of the Republic of Poland (Narodowe Centrum Nauki) through the grant number N N202 343640.}
\section*{References}
\addcontentsline{toc}{section}{References}
{%
\scriptsize
\bibliographystyle{rpkbib}
\renewcommand\refname{\vskip -1cm}
\bibliography{lqj2}           

\begin{thebibliography}{100}

\bibitem{Alberti:Uhlmann:1984}
{Alberti P.M., Uhlmann A.}, 1984, \textit{Transition probabilities on $C^*$-
  and $W^*$-algebras}, in: Baumg\"{a}rtel H., La{\ss}ner G., Pietsch A.,
  Uhlmann A. (eds.), \textit{Proceedings of the second international conference
  on operator algebras, ideals, and their applications in theoretical physics
  (Leipzig 1983)}, Teubner, Stuttgart, p.5.
  \href{http://www.physik.uni-leipzig.de/~uhlmann/PDF/Uh84a.pdf}{www.physik.uni-leipzig.de/$\sim$uhlmann/PDF/Uh84a.pdf}.

\bibitem{Ali:Silvey:1966}
{Ali S.M., Silvey S.D.}, 1966, \textit{A general class of coefficients of
  divergence of one distribution from another}, J. Roy. Stat. Soc. B
  \textbf{28}, 131.

\bibitem{Araki:1972}
{Araki H.}, 1972, \textit{Bures distance function and a generalization of
  Sakai's noncommutative Radon--Nikodym theorem}, Publ. Res. Inst. Math. Sci.
  Ky\={o}to Univ. \textbf{8}, 335.
  \href{http://dx.doi.org/10.2977/prims/1195193113}{dx.doi.org/10.2977/prims/1195193113}.

\bibitem{Araki:1973:relative:hamiltonian}
{Araki H.}, 1973, \textit{Relative hamiltonian for faithful normal states of
  von Neumann algebra}, Publ. Res. Inst. Math. Sci. Ky\={o}to Univ. \textbf{9},
  165.
  \href{http://dx.doi.org/10.2977/prims/1195192744}{dx.doi.org/10.2977/prims/1195192744}.

\bibitem{Araki:1974:modular:conjugation}
{Araki H.}, 1974, \textit{Some properties of the modular conjugation operator
  of von Neumann algebras and a non-commutative Radon--Nikodym theorem with a
  chain rule}, Pacific J. Math. \textbf{50}, 309.
  \href{http://projecteuclid.org/euclid.pjm/1102913224}{euclid:pjm/1102913224}.

\bibitem{Araki:1976:relham:relent}
{Araki H.}, 1976, \textit{Introduction to relative hamiltonian and relative
  entropy}, in: Guerra F., Robinson D.W., Stora R. (eds.), \textit{Les methodes
  mathematiques de la theorie quantique des champs, Marseille, 23-27 juin
  1975}, Colloques Internationaux C.N.R.S. \textbf{248}, \'{E}ditions du
  C.N.R.S., Paris, p.782.

\bibitem{Araki:1976:positive:cone}
{Araki H.}, 1976, \textit{Positive cone, Radon--Nikodym theorems, relative
  hamiltonian and the Gibbs condition in statistical mechanics}, in: Kastler D.
  (ed.), \textit{$C^*$-algebras and their applications to statistical mechanics
  and quantum field theory}, North-Holland, Amsterdam.

\bibitem{Araki:1976:relative:entropy:I}
{Araki H.}, 1976, \textit{Relative entropy for states of von Neumann algebras
  I}, Publ. Res. Inst. Math. Sci. Ky\={o}to Univ. \textbf{11}, 809.
  \href{http://dx.doi.org/10.2977/prims/1195191148}{dx.doi.org/10.2977/prims/1195191148}.

\bibitem{Araki:1977:relative:entropy:II}
{Araki H.}, 1977, \textit{Relative entropy for states of von Neumann algebras
  II}, Publ. Res. Inst. Math. Sci. Ky\={o}to Univ. \textbf{13}, 173.
  \href{http://dx.doi.org/10.2977/prims/1195190105}{dx.doi.org/10.2977/prims/1195190105}.

\bibitem{Araki:Masuda:1982}
{Araki H., Masuda T.}, 1982, \textit{Positive cones and $L_p$-spaces for von
  Neumann algebras}, Publ. Res. Inst. Math. Sci. Ky\={o}to Univ. \textbf{18},
  339.
  \href{http://dx.doi.org/10.2977/prims/1195183577}{dx.doi.org/10.2977/prims/1195183577}.

\bibitem{Bayes:1763}
{Bayes T.}, 1763, \textit{An essay towards solving a problem in the doctrine of
  chances}, Phil. Trans. Roy. Soc. London \textbf{53}, 370 (reprinted in: 1958,
  Biometrika \textbf{45}, 293).

\bibitem{Benoist:Marchand:1979}
{Benoist R.W., Marchand J.-P.}, 1979, \textit{Statistical inference in coupled
  quantum systems}, Lett. Math. Phys. \textbf{3}, 93.

\bibitem{Benoist:Marchand:Wyss:1979}
{Benoist R.W., Marchand J.-P., Wyss W.}, 1979, \textit{A note on relative
  entropy}, Lett. Math. Phys. \textbf{3}, 169.

\bibitem{Benoist:Marchand:Yourgrau:1977}
{Benoist R.W., Marchand J.-P., Yourgrau W.}, 1977, \textit{Statistical
  inference and quantum mechanical measurement}, Found. Phys. \textbf{7}, 827
  (addendum: 1978, Found. Phys. \textbf{8}, 117).

\bibitem{Braunstein:Caves:1988}
{Braunstein S.L., Caves C.M.}, 1988, \textit{Quantum rules: an effect can have
  more than one operation}, Found. Phys. Lett. \textbf{1}, 3.

\bibitem{Bub:1977}
{Bub J.}, 1977, \textit{Von Neumann's projection postulate as a probability
  conditionalization rule in quantum mechanics}, J. Phil. Logic \textbf{6},
  381.

\bibitem{Bub:1979}
{Bub J.}, 1979, \textit{Conditional probabilities in non-boolean possibility
  structures}, in: Hooker C.A. (ed.), \textit{The logico-algebraic approach to
  quantum mechanics}, Vol.2, Reidel, Dordrecht, p.209.

\bibitem{Bub:1979:measurement}
{Bub J.}, 1979, \textit{The measurement problem in quantum mechanics}, in:
  Toraldo di Francia G. (ed.), \textit{Problems in the foundations of physics},
  North-Holland, Amsterdam, p.71.

\bibitem{Bub:2007}
{Bub J.}, 2007, \textit{Quantum probabilities as degrees of belief}, Stud.
  Hist. Phil. Mod. Phys. \textbf{38}, 232.

\bibitem{Bures:1969}
{Bures D.J.C.}, 1969, \textit{An extension of Kakutani's theorem on infinite
  product measures to the tensor product of semiinfinite W$^*$ algebras},
  Trans. Am. Math. Soc. \textbf{135}, 194.

\bibitem{Cantoni:1975}
{Cantoni V.}, 1975, \textit{Generalized ``transition probability''}, Comm.
  Math. Phys. \textbf{44}, 125.
  \href{http://projecteuclid.org/euclid.cmp/1103899296}{euclid:cmp/1103899296}.

\bibitem{Carazza:Casartelli:DElia:1977}
{Carazza B., Casartelli M., D'Elia A.}, 1977, \textit{Segal entropy and the
  principle of least interference}, Phys. Lett. A \textbf{62}, 205.

\bibitem{Caticha:Giffin:2006}
{Caticha A., Giffin A.}, 2006, \textit{Updating probabilities}, in:
  Mohammad-Djafari A. (ed.), \textit{Bayesian inference and maximum entropy
  methods in science and engineering}, AIP Conf. Proc. \textbf{872}, 31.
  \href{http://www.arxiv.org/pdf/physics/0608185}{arXiv:physics/0608185}.

\bibitem{Caves:Fuchs:Schack:2001}
{Caves C.M., Fuchs C.A., Schack R.}, 2001, \textit{Quantum probabilities as
  bayesian probabilities}, Phys. Rev. A \textbf{65}, 022305.
  \href{http://www.arxiv.org/pdf/quant-ph/0106133}{arXiv:quant-ph/0106133}.

\bibitem{Coles:2012}
{Coles P.J.}, 2012, \textit{Unification of different views of decoherence and
  discord}, Phys. Rev. A \textbf{85}, 042103.
  \href{http://www.arxiv.org/pdf/1110.1664}{arXiv:1110.1664}.

\bibitem{Connes:1973:classification}
{Connes A.}, 1973, \textit{Une classification des facteurs de type III}, Ann.
  Sci. \'{E}cole Norm. Sup. 4 \`{e}me s\'{e}r. \textbf{6}, 133.
  \href{http://archive.numdam.org/article/ASENS_1973_4_6_2_133_0.pdf}{numdam:ASENS\_1973\_4\_6\_2\_133\_0}.

\bibitem{Connes:1974}
{Connes A.}, 1974, \textit{Caract\'{e}risation des espaces vectoriels
  ordonn\'{e}s sous-jacents aux alg\`{e}bres de von Neumann}, Ann. Inst.
  Fourier Grenoble \textbf{24}, 121.
  \href{http://archive.numdam.org/article/AIF_1974__24_4_121_0.pdf}{numdam:AIF\_1974\_\_24\_4\_121\_0}.

\bibitem{Csiszar:1963}
{Csisz\'{a}r I.}, 1963, \textit{Eine informationstheoretische Ungleichung und
  ihre Anwendung auf den Beweis der Ergodizit\"{a}t von Markoffschen Ketten},
  Magyar Tud. Akad. Mat. Kutat\'{o} Int. K\"{o}zl. \textbf{8}, 85.

\bibitem{Laplace:1774}
{de Laplace P.-S.}, 1774, \textit{M\'{e}moire sur la probabilit\'{e} des causes
  par les \'{e}v\'{e}nements}, M\'{e}m. Acad. Roy. Sci. Paris (Savants
  \'{E}trangers) \textbf{6}, 621. (reprinted in: 1986, Stat. Sci. \textbf{1},
  359.).

\bibitem{Laplace:1812}
{de Laplace P.-S.}, 1812, \textit{Th\'{e}orie analytique des probabilit\'{e}s},
  Courcier, Paris.

\bibitem{Laplace:1814}
{de Laplace P.-S.}, 1814, \textit{Essai philosophique sur les
  probabilit\'{e}s}, Courcier, Paris. (Engl. transl. 1951, \textit{A
  philosophical essay on probabilities}, Dover, New York).

\bibitem{deMuynck:2002}
{de Muynck W.M.}, 2002, \textit{Foundations of quantum mechanics: an empiricist
  approach}, Kluwer, Dordrecht.

\bibitem{Diaconis:Zabell:1982}
{Diaconis P., Zabell S.}, 1982, \textit{Updating subjective probability}, J.
  Amer. Statist. Assoc. \textbf{77}, 822.

\bibitem{Dieks:Veltkamp:1983}
{Dieks D., Veltkamp P.}, 1983, \textit{Distance between quantum states,
  statistical inference and the projection postulate}, Phys. Lett. A
  \textbf{97}, 24.

\bibitem{Digernes:1975}
{Digernes T.}, 1975, \textit{Duality for weights on covariant systems and its
  applications}, Ph.D. thesis, University of California, Los Angeles.

\bibitem{Diu:1982}
{Diu B.}, 1982, \textit{Note on a recent proposal concerning statistical
  inference in quantum theory}, Ann. Inst. Henri Poincar\'{e} A \textbf{37},
  59.

\bibitem{Diu:1983}
{Diu B.}, 1983, \textit{Statistical inference and distance between states in
  quantum mechanics}, Ann. Inst. Henri Poincar\'{e} A \textbf{38}, 167.

\bibitem{Domotor:1985}
{Domotor Z.}, 1985, \textit{Probability kinematics, conditionals, and entropy
  principles}, Synthese \textbf{63}, 75.

\bibitem{Domotor:Zanotti:Graves:1980}
{Domotor Z., Zanotti H., Graves H.}, 1980, \textit{Probability kinematics},
  Synthese \textbf{44}, 421.

\bibitem{Donald:1986}
{Donald M.J.}, 1986, \textit{On the relative entropy}, Commun. Math. Phys.
  \textbf{105}, 13.
  \href{http://projecteuclid.org/euclid.cmp/1104115254}{euclid:cmp/1104115254}.

\bibitem{Donald:1987:further:results}
{Donald M.J.}, 1987, \textit{Further results on the relative entropy}, Math.
  Proc. Cambridge Phil. Soc. \textbf{101}, 363.

\bibitem{Donald:1990}
{Donald M.J.}, 1990, \textit{Relative hamiltonians which are not bounded from
  above}, J. Funct. Anal. \textbf{91}, 143.

\bibitem{Douven:Romeijn:2012}
{Douven I., Romeijn J.-W.}, 2012, \textit{A new resolution of the Judy Benjamin
  problem}, Mind \textbf{479}, 637.
  \href{http://eprints.lse.ac.uk/27004/1/A_new_resoltuion_(LSERO).pdf}{eprints.lse.ac.uk/27004/1/A\_new\_resoltuion\_(LSERO).pdf}.

\bibitem{Field:1978}
{Field H.}, 1978, \textit{A note on Jeffrey conditionalization}, Phil. Sci.
  \textbf{45}, 361.

\bibitem{Fine:1973}
{Fine T.L.}, 1973, \textit{Theories of probability}, Academic Press, New York.

\bibitem{Fok:1932:nachala}
{Fok V.A.}, 1932, \textit{Nachala kvantovoi mekhaniki}, Kubuch, Leningrad
  (Engl. transl. 1978, Fundamentals of quantum mechanics, Mir, Moskva).

\bibitem{Frechet:1906}
{Fr\'{e}chet M.}, 1906, \textit{Sur quelques points du calcul fonctionnel},
  Rend. Circ. Matem. Palermo \textbf{22}, 1.
  \href{http://webpages.ursinus.edu/nscoville/Frechet Thesis French
  I.pdf}{webpages.ursinus.edu/nscoville/Frechet Thesis French I.pdf},
  \href{http://webpages.ursinus.edu/nscoville/Frechet Thesis French
  II.pdf}{webpages.ursinus.edu/nscoville/Frechet Thesis French II.pdf},
  \href{http://webpages.ursinus.edu/nscoville/Frechet Thesis French
  III.pdf}{webpages.ursinus.edu/nscoville/Frechet Thesis French III.pdf},
  \href{http://webpages.ursinus.edu/nscoville/Frechet Thesis French
  IV.pdf}{webpages.ursinus.edu/nscoville/Frechet Thesis French IV.pdf}.

\bibitem{Fuchs:2002}
{Fuchs C.A.}, 2002, \textit{Quantum mechanics as quantum information (and only
  a little more)},
  \href{http://www.arxiv.org/pdf/quant-ph/0205039}{arXiv:quant-ph/0205039}.

\bibitem{Fuchs:2003}
{Fuchs C.A.}, 2003, \textit{Quantum mechanics as quantum information, mostly},
  J. Mod. Opt. \textbf{50}, 987.
  \href{http://www.arxiv.org/pdf/quant-ph/0205039}{arXiv:quant-ph/0205039}.

\bibitem{Gelfand:Naimark:1943}
{Gel'fand I.M., Na\u{\i}mark M.A.}, 1943, \textit{On the imbedding of normed
  rings into the ring of operators in Hilbert space}, Matem. Sbornik (N.S.)
  \textbf{12}, 197. \href{http://mi.mathnet.ru/msb6155}{mathnet.ru:msb6155}.

\bibitem{Giffin:2008}
{Giffin A.}, 2008, \textit{Maximum entropy: the universal method for
  inference}, Ph.D. thesis, State University of New York, Albany.
  \href{http://www.arxiv.org/pdf/0901.2987}{arXiv:0901.2987}.

\bibitem{Giffin:Caticha:2006}
{Giffin A., Caticha A.}, 2007, \textit{Updating probabilities with data and
  moments}, in: Knuth K. et al. (eds.), \textit{Bayesian inference and maximum
  entropy methods in science and engineering}, AIP Conf. Proc. \textbf{954},
  74. \href{http://www.arxiv.org/pdf/0708.1593}{arXiv:0708.1593}.

\bibitem{Good:1950}
{Good I.J.}, 1950, \textit{Probability and the weighing of evidence}, Griffin,
  London.

\bibitem{Gudder:1980}
{Gudder S.P.}, 1980, \textit{Statistical inference in quantum mechanics}, Rep.
  Math. Phys. \textbf{17}, 265.

\bibitem{Gudder:Marchand:1977}
{Gudder S.P., Marchand J.-P.}, 1977, \textit{Conditional expectations on von
  Neumann algebras: a new approach}, Rep. Math. Phys. \textbf{12}, 317.

\bibitem{Gudder:Marchand:Wyss:1979}
{Gudder S.P., Marchand J.-P., Wyss W.}, 1979, \textit{Bures distance and
  relative entropy}, J. Math. Phys. \textbf{20}, 1963.

\bibitem{Haagerup:1973}
{Haagerup U.}, 1973, \textit{The standard form of von Neumann algebras},
  Preprint Ser. 1973 No. \textbf{15}, K{\o}benhavns Universitet Matematisk
  Institut, K{\o}benhavn.

\bibitem{Haagerup:1975:standard:form}
{Haagerup U.}, 1975, \textit{The standard form of von Neumann algebras}, Math.
  Scand. \textbf{37}, 271.
  \href{http://www.mscand.dk/article.php?id=2275}{www.mscand.dk/article.php?id=2275}.

\bibitem{Hadjisavvas:1978}
{Hadjisavvas N.}, 1978, \textit{\'{E}tude de certaines consequences d'une
  interpr\'{e}tation subjective de la notion d'\'{e}tat}, Ann. Fond. Louis de
  Broglie \textbf{3}, 155.

\bibitem{Hadjisavvas:1981}
{Hadjisavvas N.}, 1981, \textit{Distance between states and statistical
  inference in quantum theory}, Ann. Inst. Henri Poincar\'{e} A \textbf{38},
  167.
  \href{http://archive.numdam.org/article/AIHPA_1981__35_4_287_0.pdf}{numdam:AIHPA\_1981\_\_35\_4\_287\_0.pdf}.

\bibitem{Hasegawa:1993}
{Hasegawa H.}, 1993, \textit{$\alpha$-divergence of the non-commutative
  information geometry}, Rep. Math. Phys. \textbf{33}, 87.

\bibitem{HKK:2014}
{Hellmann F., Kami\'{n}ski W., Kostecki R.P.}, 2014, \textit{Quantum collapse
  rules from the maximum relative entropy principle},
  \href{http://arxiv.org/pdf/1407.7766}{arXiv:1407.7766}.

\bibitem{Henderson:2010}
{Henderson L.}, 2010, \textit{Bayesian updating and information gain in quantum
  measurement}, in: Bokulich A., Jaeger G. (eds.), \textit{Philosophy of
  quantum information and entanglement}, Cambridge University Press, Cambridge,
  p.151.

\bibitem{Herbut:1969}
{Herbut F.}, 1969, \textit{Derivation of the change of state in measurement
  from the concept of minimal measurement}, Ann. Phys. \textbf{55}, 271.

\bibitem{Hobson:Cheng:1973}
{Hobson A., Cheng B.-K.}, 1973, \textit{A comparison of the Shannon and
  Kullback information measures}, J. Stat. Phys. \textbf{7}, 301.

\bibitem{Hughes:vanFraassen:1984}
{Hughes R.I.G., van Fraassen B.C.}, 1984, \textit{Symmetry arguments in
  probability kinematics}, in: Kitcher P., Asquith P. (ed.),
  \textit{Proceedings of the biennal meeting of the Philosophy of Science
  Association}, Vol.2, Philosophy of Science Association, East Lausing, p.851.

\bibitem{Jacobs:2002}
{Jacobs K.}, 2002, \textit{How do two observers pool their knowledge about a
  quantum system?}, Quant. Inf. Proc. \textbf{1}, 73.
  \href{http://www.arxiv.org/pdf/quant-ph/0201096}{arXiv:quant-ph/0201096}.

\bibitem{Jamison:1974}
{Jamison B.}, 1974, \textit{A Martin boundary interpretation of a minimum
  change principle to probability kinematics}, Z. Warschein. Geb. \textbf{30},
  265.

\bibitem{Jauch:Misra:Gibson:1968}
{Jauch J.M., Misra B., Gibson A.G.}, 1968, \textit{On the asymptotic condition
  of scattering theory}, Helv. Phys. Acta \textbf{41}, 513.

\bibitem{Jeffrey:1957}
{Jeffrey R.C.}, 1957, \textit{Contributions to the theory of inductive
  probability}, Ph.D. thesis, Princeton University, Princeton.

\bibitem{Jeffrey:1965}
{Jeffrey R.C.}, 1965, \textit{The logic of decision}, Chicago University Press,
  Chicago.

\bibitem{Jeffrey:1968}
{Jeffrey R.C.}, 1968, \textit{Probable knowledge}, in: Lakatos I. (ed.),
  \textit{The problem of inductive logic}, North-Holland, Amsterdam, p.166.

\bibitem{Jencova:2005}
{Jen\v{c}ov\'{a} A.}, 2005, \textit{Quantum information geometry and
  non-commutative $L_p$ spaces}, Inf. Dim. Anal. Quant. Prob. Relat. Top.
  \textbf{8}, 215.
  \href{http://www.mat.savba.sk/~jencova/lpspaces.pdf}{www.mat.savba.sk/$\sim$jencova/lpspaces.pdf}.

\bibitem{Johnson:1979}
{Johnson R.W.}, 1979, \textit{Axiomatic characterization of the directed
  divergences and their linear combinations}, IEEE Trans. Inf. Theory
  \textbf{25}, 709.

\bibitem{Kemble:1937}
{Kemble E.C.}, 1937, \textit{The fundamental principles of quantum mechanics},
  McGraw--Hill, New York.

\bibitem{Kosaki:1980:PhD}
{Kosaki H.}, 1980, \textit{Canonical $L^p$-spaces associated with an arbitrary
  abstract von Neumann algebra}, Ph.D. thesis, University of California, Los
  Angeles.
  \href{http://dmitripavlov.org/scans/kosaki-thesis.pdf}{dmitripavlov.org/scans/kosaki-thesis.pdf}.

\bibitem{Kostecki:2011:OSID}
{Kostecki R.P.}, 2011, \textit{The general form of $\gamma$-family of quantum
  relative entropies}, Open Sys. Inf. Dyn. \textbf{18}, 191.
  \href{http://www.arxiv.org/pdf/1106.2225}{arXiv:1106.2225}.

\bibitem{Kostecki:2011:Waterloo:talk}
{Kostecki R.P.}, 2011, \textit{Quantum information geometric foundations of
  quantum theory and space-time}, talk given at Quantum foundations seminar
  (July 12), Perimeter Institute, Waterloo.

\bibitem{Kostecki:2013}
{Kostecki R.P.}, 2013, \textit{$W^*$-algebras and noncommutative integration},
  \href{http://www.arxiv.org/pdf/1307.4818}{arXiv:1307.4818}.

\bibitem{Kostecki:2014:towards}
{Kostecki R.P.}, 2014, \textit{Towards quantum information geometric
  foundations}, in preparation.

\bibitem{Kronfli:1970}
{Kronfli N.S.}, 1970, \textit{States on generalised logics}, Int. J. Theor.
  Phys. \textbf{3}, 191.

\bibitem{Kullback:1959}
{Kullback S.}, 1959, \textit{Information theory and statistics}, Wiley, New
  York (2nd ed. 1968).
  \href{http://libgen.org/get?open=0&md5=66875BFC74B54B9B73E51375EF11FEDB}{libgen.org:66875BFC74B54B9B73E51375EF11FEDB}.

\bibitem{Kullback:Leibler:1951}
{Kullback S., Leibler R.A.}, 1951, \textit{On information and sufficiency},
  Ann. Math. Statist. \textbf{22}, 79.
  \href{http://projecteuclid.org/euclid.aoms/1177729694}{euclid:aoms/1177729694}.

\bibitem{Lueders:1951}
{L\"{u}ders G.}, 1951, \textit{\"{U}ber die Zustands\"{a}nderung durch den
  Messprozess}, Ann. Phys. Leipzig \textbf{8}, 322. (Engl. transl.: 2004,
  \textit{Concerning the state-change due to the measurement process},
  \href{http://www.arxiv.org/pdf/quant-ph/0403007}{arXiv:quant-ph/0403007}).

\bibitem{Marchand:1977}
{Marchand J.-P.}, 1977, \textit{Relative coarse-graining}, Found. Phys.
  \textbf{7}, 35.

\bibitem{Marchand:1981}
{Marchand J.-P.}, 1981, \textit{Statistical inference in quantum mechanics},
  in: Gustafson K.E. et al (eds.), \textit{Quantum mechanics in mathematics,
  chemistry and physics}, Plenum, New York, p.73.

\bibitem{Marchand:1983}
{Marchand J.-P.}, 1983, \textit{Statistical inference by minimal Bures
  distance}, in: van der Merwe (ed.), \textit{Old and new questions in physics,
  cosmology, philosophy, and theoretical biology. Essays in honor of Wolfgang
  Yourgreau}, Plenum, New York, p.275.
  \href{http://libgen.org/get?open=0&md5=e97f930eac263d74a4c065bbc15287ae}{libgen.org:e97f930eac263d74a4c065bbc15287ae}.

\bibitem{Marchand:1983:Milano}
{Marchand J.-P.}, 1983, \textit{Statistical inference in non-commutative
  probability}, Rend. Sem. Math. Fis. Milano \textbf{52}, 551.

\bibitem{Marchand:Wyss:1977}
{Marchand J.-P., Wyss W.}, 1977, \textit{Statistical inference and entropy}, J.
  Stat. Phys. \textbf{16}, 349.

\bibitem{Margenau:1936}
{Margenau H.}, 1936, \textit{Quantum mechanical descriptions}, Phys. Rev.
  \textbf{49}, 240.

\bibitem{Margenau:1963}
{Margenau H.}, 1963, \textit{Measurements and quantum states: part I, II},
  Phil. Sci. \textbf{30}, 1, 138.

\bibitem{Masuda:1984}
{Masuda T.}, 1984, \textit{A note on a theorem of A.~Connes on Radon-Nikodym
  cocycles}, Publ. Res. Inst. Math. Sci. Ky\={o}to Univ. \textbf{20}, 131.
  \href{http://dx.doi.org/10.2977/prims/1195181833}{doi:10.2977/prims/1195181833}.

\bibitem{May:1973}
{May S.J.}, 1973, \textit{On the application of a minimum change principle to
  probability kinematics}, Ph.D. thesis, University of Waterloo, Waterloo.

\bibitem{May:1976}
{May S.J.}, 1976, \textit{Probability kinematics: a constrained optimization
  problem}, J. Phil. Log. \textbf{5}, 395.

\bibitem{May:1979}
{May S.J.}, 1979, \textit{An application of Neustadt's abstract maximum
  principle to probability kinematics}, J. Optim. Th. Appl. \textbf{27}, 249.

\bibitem{May:Harper:1976}
{May S.J., Harper W.L.}, 1976, \textit{Towards an optimisation procedure for
  applying minimum change principles in probability kinematics}, in: Harper
  W.L., Hooker C.A. (eds.), \textit{Foundations of probability theory,
  statistical inference, and statistical theories of science}, Vol.1, Reidel,
  Dordrecht, p.137.

\bibitem{MPSVW:2010}
{Modi K., Paterek T., Son W., Vedral V., Williamson M.}, 2010, \textit{Unified
  view of quantum and classical correlations}, Phys. Rev. Lett. \textbf{104},
  080501. \href{http://www.arxiv.org/pdf/0911.5417}{arXiv:0911.5417}.

\bibitem{Morimoto:1963}
{Morimoto T.}, 1963, \textit{Markov processes and the $H$-theorem}, J. Phys.
  Soc. Jap. \textbf{12}, 328.

\bibitem{Murray:vonNeumann:1936}
{Murray F.J., von Neumann J.}, 1936, \textit{On rings of operators}, Ann. Math.
  \textbf{37}, 116.

\bibitem{Ojima:2004}
{Ojima I.}, 2004, \textit{Temperature as order parameter of broken scale
  invariance}, Publ. Res. Inst. Math. Sci. Ky\={o}to Univ. \textbf{40}, 731.
  \href{http://arxiv.org/pdf/math-ph/0311025}{arXiv:math-ph/0311025}.

\bibitem{Palge:Konrad:2008}
{Palge V., Konrad T.}, 2008, \textit{A remark on Fuchs' bayesian interpretation
  of quantum mechanics}, Stud. Hist. Phil. Mod. Phys. \textbf{39}, 273.

\bibitem{Petz:1985:properties}
{Petz D.}, 1985, \textit{Properties of quantum entropy}, in: Accardi L., von
  Waldenfels W. (eds.), \textit{Quantum probability and applications II}, LNM
  \textbf{1136}, Springer, Berlin, p.428.

\bibitem{Petz:1986:properties}
{Petz D.}, 1986, \textit{Properties of the relative entropy of states of von
  Neumann algebra}, Acta Math. Hungar. \textbf{47}, 65.

\bibitem{Raggio:1982}
{Raggio G.A.}, 1982, \textit{Comparison of Uhlmann's transition probability
  with the one induced by the natural positive cone of von Neumann algebras in
  standard form}, Lett. Math. Phys. \textbf{6}, 233.

\bibitem{Raggio:1984}
{Raggio G.A.}, 1984, \textit{Generalized transition probabilities and
  applications}, in: Accardi L. et al (eds.), \textit{Quantum probability and
  appplications to quantum theory of irreversible processes}, Springer, Berlin,
  p.327.
  \href{http://libgen.org/get?open=0&md5=c4a12cdc49b77c7d0d999f6bfb49c1c1}{libgen.org:c4a12cdc49b77c7d0d999f6bfb49c1c1}.

\bibitem{Redei:1992:Bayes}
{R\'{e}dei M.}, 1992, \textit{When can non-commutative statistical inference be
  bayesian?}, Int. Stud. Phil. Sci. \textbf{6}, 129.

\bibitem{Redei:Summers:2007}
{R\'{e}dei M., Summers S.J.}, 2007, \textit{Quantum probability theory}, Stud.
  Hist. Phil. Mod. Phys. \textbf{38}, 390.
  \href{http://www.arxiv.org/pdf/quant-ph/0601158}{arXiv:quant-ph/0601158}.

\bibitem{Sakai:1965}
{Sakai S.}, 1965, \textit{A Radon--Nikodym theorem in $W^*$-algebras}, Bull.
  Amer. Math. Soc. \textbf{71}, 149.
  \href{http://projecteuclid.org/euclid.bams/1183526404}{euclid:bams/1183526404}.

\bibitem{Schack:Brun:Caves:2001}
{Schack R., Brun T.A., Caves C.M.}, 2001, \textit{Quantum Bayes rule}, Phys.
  Rev. A \textbf{64}, 014305.
  \href{http://www.arxiv.org/pdf/quant-ph/0008113}{arXiv:quant-ph/0008113}.

\bibitem{Schwinger:1959}
{Schwinger J.}, 1959, \textit{The algebra of microscopic measurement}, Proc.
  Nat. Acad. Sci. U.S.A. \textbf{45}, 1542.

\bibitem{Segal:1947:irreducible}
{Segal I.E.}, 1947, \textit{Irreducible representations of operator algebras},
  Bull. Amer. Math. Soc. \textbf{61}, 69.
  \href{http://www.ams.org/journals/bull/1947-53-02/S0002-9904-1947-08742-5/S0002-9904-1947-08742-5.pdf}{www.ams.org/journals/bull/1947-53-02/S0002-9904-1947-08742-5/S0002-9904-1947-08742-5.pdf}.

\bibitem{Streater:2007}
{Streater R.F.}, 2007, \textit{Lost causes in and beyond physics}, Springer,
  Berlin.

\bibitem{Tribus:Rossi:1973}
{Tribus M., Rossi R.}, 1973, \textit{On the Kullback information measure as a
  basis for information theory: comments on a proposal by Hobson and Chang}, J.
  Stat. Phys. \textbf{9}, 331.

\bibitem{Uhlmann:1976}
{Uhlmann A.}, 1976, \textit{The ``transition probability'' in the state space
  of a $^*$-algebra}, Rep. Math. Phys. \textbf{9}, 273.
  \href{http://www.physik.uni-leipzig.de/~uhlmann/PDF/Uh76a.pdf}{www.physik.uni-leipzig.de/$\sim$uhlmann/PDF/Uh76a.pdf}.

\bibitem{Umegaki:1961}
{Umegaki H.}, 1961, \textit{On information in operator algebras}, Proc. Jap.
  Acad. \textbf{37}, 459.
  \href{http://projecteuclid.org/euclid.pja/1195523632}{euclid:pja/1195523632}.

\bibitem{Umegaki:1962}
{Umegaki H.}, 1962, \textit{Conditional expectation in an operator algebra, IV
  (entropy and information)}, K\={o}dai Math. Sem. Rep. \textbf{14}, 59.
  \href{http://projecteuclid.org/euclid.kmj/1138844604}{euclid:kmj/1138844604}.

\bibitem{Valente:2007}
{Valente G.}, 2007, \textit{Is there a stability problem for bayesian
  noncommutative probabilities?}, Stud. Hist. Phil. Mod. Phys. \textbf{38},
  832.

\bibitem{vanFraassen:1981}
{van Fraassen B.C.}, 1981, \textit{A problem for relative information
  minimizers in probability kinematics}, Brit. J. Phil. Sci. \textbf{32}, 375.

\bibitem{vonNeumann:1930:algebra}
{von Neumann J.}, 1930, \textit{Zur algebra der Funktionaloperatoren und
  Theorie der normalen Operatoren}, Math. Ann. \textbf{102}, 370.
  \href{http://gdz.sub.uni-goettingen.de/dms/load/img/?PPN=GDZPPN002273675&IDDOC=38466}{gdz.sub.uni-goettingen.de/dms/load/img/?PPN=GDZPPN002273675\&IDDOC=38466}.

\bibitem{vonNeumann:1932:grundlagen}
{von Neumann J.}, 1932, \textit{Mathematische Grundlagen der Quantenmechanik},
  Springer, Berlin.
  \href{http://libgen.org/get?open=0&md5=DF8F17426E6D36B4AD2D350970158BD7}{libgen.org:DF8F17426E6D36B4AD2D350970158BD7}
  (Engl. transl. 1955, \textit{Mathematical foundations of quantum mechanics},
  Princeton University Press, Princeton.
  \href{http://libgen.org/get?open=0&md5=C23114E46AB9284F70E789F93BC0512D}{libgen.org:C23114E46AB9284F70E789F93BC0512D}).

\bibitem{Wald:1947}
{Wald A.}, 1947, \textit{Sequential analysis}, Wiley, New York.

\bibitem{Warmuth:2005}
{Warmuth M.K.}, 2005, \textit{A Bayes rule for density matrices}, in: Weiss Y.,
  Sch\"{o}lkopf B., Platt J. (eds.), \textit{Advances in neural information
  processing systems \textbf{18} (NIPS 05)}, MIT Press, p.1457.
  \href{http://users.soe.ucsc.edu/~manfred/pubs/C72.pdf}{users.soe.ucsc.edu/$\sim$manfred/pubs/C72.pdf}.

\bibitem{Williams:1980}
{Williams P.M.}, 1980, \textit{Bayesian conditionalisation and the principle of
  minimum information}, Brit. J. Phil. Sci. \textbf{31}, 131.

\bibitem{Yamagami:2008}
{Yamagami S.}, 2008, \textit{Geometric mean of states and transition
  amplitudes}, Lett. Math. Phys. \textbf{84}, 123.
  \href{http://www.arxiv.org/pdf/0801.0858}{arXiv:0801.0858}.

\bibitem{Yamagami:2010}
{Yamagami S.}, 2010, \textit{Geometry of quasi-free states of CCR algebras},
  Int. J. Math. \textbf{21}, 875.
  \href{http://www.arxiv.org/pdf/0801.1739}{arXiv:0801.1739}.

\bibitem{Zellner:1988}
{Zellner A.}, 1988, \textit{Optimal information processing and Bayes' theorem},
  Amer. Stat. \textbf{42}, 278.

\end{thebibliography}
}%
\end{document}